\documentclass[11pt,a4paper]{article}

\usepackage{hyperref}
\hypersetup{colorlinks=true,citecolor=black,linkcolor=black,urlcolor=black,pdfview=FitH,pdfstartview=FitH}

\usepackage{mystyle}

\usepackage{times}
\renewcommand{\forall}{\textrm{ for all }}
\usepackage{tocbibind}
\usepackage{natbib}
\bibpunct{(}{)}{,}{a}{}{;} 

\usepackage{mathrsfs}

\usepackage{amsthm}
\newtheorem{lemma}{Lemma}
\newtheorem{property}{Property}

\usepackage{todonotes}
\presetkeys{todonotes}{color=yellow!20}{}

\usepackage{tikz}
\usetikzlibrary{arrows}
\usetikzlibrary{decorations.text}
\usetikzlibrary{positioning}
\usetikzlibrary{calc}

\usepackage{pgfplots}
\pgfplotsset{compat=1.14}

\usepackage{amsmath,amsfonts}

\DeclareMathOperator*{\argminlex}{{argmin^{\it lex}}}

\usepackage[linesnumbered,lined,boxed,commentsnumbered]{algorithm2e}
\SetKwInOut{Input}{Input}
\SetKw{kwProblem}{Problem:}
\SetKwProg{Fn}{Function}{}{end}
\SetKw{KwBy}{increment by}
\SetEndCharOfAlgoLine{}

\usetikzlibrary{shapes.misc}
\tikzset{cross/.style={cross out, draw=black, minimum size=2*(#1-\pgflinewidth), inner sep=0pt, outer sep=0pt},
cross/.default={1pt}}

\usepackage{subcaption}

\newcommand{\range}[2]{\{{#1}, \ldots, {#2}\}}

\Author{Steffen Pottel and Asvin Goel}{Kühne Logistics University, 20457 Hamburg, Germany}

\date{\today}
\Title{Scheduling activities with time-dependent durations and resource consumptions}

\begin{document}
\maketitle
\begin{abstract}
In this paper we study time-dependent scheduling problems where activities consume a resource with limited availability.
Activity durations as well as resource consumptions are assumed to be time-dependent and the resource can be replenished between activities.
Because of the interaction of time-dependent activity durations and resource consumptions,
scheduling policies based on starting all activities as early as possible may fail due to unnecessarily high resource consumptions.
We propose a dynamic discretization discovery algorithm that generates a partially time-expanded network during the search.
We propose preloading techniques allowing to significantly reduce the computational effort if the approach is embedded in an iterative solution procedure that frequently evaluates activity sequences that start with the same activities.
We evaluate our approaches on a case of routing a fleet of electric vehicles in which vehicles can recharge batteries during the route.
\ \\[0.5cm]
\end{abstract}

\section{Introduction}
Many operational activities consume some kind of resource throughout their execution.
For example, engines used for executing activities consume energy.
With the increasing replacement of internal combustion engines by battery electric propulsion systems, the limited battery capacity has to be considered when scheduling activities and it must be ensured that all activities can be executed without running out of energy \citep[see, e.g.,][]{PJL2016}.
Similarly, the time required by a human operator can be seen as a resource that is consumed throughout the execution of activities and which is limited by work regulations \citep[see, e.g.,][]{FMCSA11, EU03}.
Also machine times may be limited between preventive maintenance processes or inspection intervals, for example, aircrafts must be inspected before a total flight time of a given limit is reached \citep{LII}.

In this paper we study the problem of scheduling a sequence of activities where
activities consume a limited resource throughout execution.
We assume that both activity durations and resource consumptions depend on the time of execution.
A major cause of time-dependent activity durations and resource consumption is congestion, which can have a significant impact on the duration required to complete certain activities.
In the~15~largest urban areas in the United States, for example, travel times on freeways are 44~percent higher during peak periods than during free flow times \citep{urbanmobilityreport2019}.
Similarly, in warehouses and manufacturing facilities, congestion can occur due to space limitations and narrow aisles \citep{ZBN2009}.
Thus the time required for picking activities or moving a forklift may vary over the course of the day.
Time-dependent durations may also occur for other reasons.
For example, the total trip duration of a public transport user includes waiting time until the next scheduled departure and the actual travel time.
Thus, the total trip duration depends on the time of the day and can be represented as time-dependent function \citep{RH2016}.
In satellite imagery, as another example,
 the relative position of the satellite to an object on the earth surface depends on the position of the satellite in the orbit and the time required to tilt the satellite accordingly can be represented by a time-dependent function \citep{LLCH2017}.

Time-dependent activity durations can have various ramifications that have received very little attention in the scientific literature.
Consider, for example, the case of electric forklifts or vehicles.
If activity durations and energy consumptions are time-dependent, the timing of the activities can be of particular importance in order to ensure that the total energy consumption of a sequence of activities stays below the total energy provided by the battery.
Similarly, in the case of limited working times,
the total working time of an operator may exceed the maximum working time allowed if most activities are conducted during peak-hours.
Thus, activities may have to be shifted away from peak-hours in order not to exceed the legal limits on the working time.
Therefore, the commonly used scheduling policy to schedule activities as early as possible may fail due to unnecessarily high resource consumptions.

Sometimes resources can be regenerated after they have been consumed.
For example, an electric battery can be recharged, a human operator can take a rest as required by working regulations, or a maintenance process can reset the machine time available until the next maintenance is needed.
In order to accommodate for such cases, we also study the problem of jointly scheduling activities and replenishments.

The contributions of this paper are:
1) we introduce the time-dependent activity scheduling problem with resource constraints;
2) we show that an optimal solution can easily be found if activity durations fulfil the the so-called \emph{first in first out (FIFO)} property, i.e the property that an earlier start time implies an earlier or equal completion time, and resource consumption functions are non-decreasing in time;
3) we propose a dynamic discretization discovery algorithm for solving the problem in the more general case of non-monotonous consumption functions;
4) we extend the approach to the case where resources can be replenished between activities;
5) we propose an acceleration technique for iterative problem solving  by preloading vertices of the partially expanded network obtained from a previous iteration;
6) we propose new benchmark instances for time-dependent vehicle routing with electric vehicles;
7) we evaluate the algorithms on these benchmark instances.

\section{Related work}

Resources that are consumed throughout the course of activities can be found in various scheduling problems.
Often a capacity limit is given on such resources and the total consumption must not exceed the capacity.
In the case where each activity has a constant resource consumption, the total resource consumption can easily be obtained by summing up the constants.
This is for example the case in the capacitated vehicle routing problem \citep[see, e.g.,][]{ITV2014}
where the maximum load of the vehicle can be interpreted as a resource that is depleted by the customer demands.
In the multi-trip vehicle routing problem \citep[see, e.g.,][]{Fle1990,TLG96,CAFV2014}, this full capacity of the resource can be restored when the vehicle returns to the depot.
\cite{SSWL2019} provide a recent review of vehicle routing problems in which vehicle routes include stops for replenishments, refueling, and breaks.

For route planning problems with battery-powered vehicles \citep[see, e.g.,][]{PJL2016}, the energy consumption of the vehicles depends on the distance travelled.
Therefore, minimising route distances will help to create routes that can be conducted without running out of energy.
Recharging the battery may be possible in order to increase the range of a vehicle.
\cite{SSG2014} study a vehicle routing problem in which electric vehicles can visit charging stations during the routes to recharge the battery.
Unlike \cite{SSG2014} who assume that the battery is always fully recharged, \cite{FORT2014} allow to partially recharge the battery.
\cite{SDK2017} find optimal recharging policies for a given route determining when to recharge as well as the amount of energy charged.
\cite{MONTOYA201787} consider non-linear charging functions and optimise where and which amount to charge.
\cite{BDGWZ2019} find optimal routes on road networks for a realistic distribution of charging stations as well as non-linear charging functions allowing partial recharging.

In the vehicle routing and truck driver scheduling problem  \citep[see, e.g.,][]{TiGo2020,GV14_VRTDSP,PDDR10, Go09_VRTDSP},
vehicle routes have to be generated in such a way that all customers can be visited within time windows and all routes can be executed without violating government regulations with respect to driving time limits.
According to these regulations, truck drivers have to take breaks and rest periods at regular intervals.
The remaining time without a break or rest period can be regarded as a resource that is consumed when driving.
Although the driving time between two customers is assumed to be a given constant, time windows for customer visits may make it impossible to always drive as long as allowed by the regulations and only take breaks or rests when the legal driving time limits are reached.
Instead, determining feasibility  of each route requires the solution of a truck driver scheduling problem \citep[see, e. g.,][]{AS09, Go13_US2013TDSP, Go10_EUTDSP, GK12_EUTeamTDSP, GR12_CANTDSP, GAS12_AUSTDSP}.
Minimising schedule durations is an even more complicated problem that is studied, for example, by \cite{KHS2011}, \cite{Go12_MDTDSP,Go12_CANMDTDSP,Go12_AUSMDTDSP}, and \cite{RCL13}.

Most of the above-mentioned studies assume that activity durations 
are given as constant values.
In many real-life applications, however, the duration 
of activities depends on the time of the day.
Scheduling problems with time-dependent processing times have been surveyed by
\cite{CDL2004} and \cite{AW1999}.
\cite{VCGP2015} review activity scheduling problems in which timing decisions are of particular importance and possibly time-dependent.
\cite{SCM2008} study a scheduling problem where setups with a time dependent duration are required between pairs of jobs and formulate the problem as a time-dependent traveling salesman problem.
The  time-dependent traveling salesman problem is also studied, e.g.,  by \cite{NDSL2012}, \cite{MMM2017}, \cite{CGG2014}, \cite{TGJL2016}.
\cite{VHBS2019}  apply dynamic discretization discovery, originally proposed by \cite{BHMS2017}, for the time-dependent traveling salesman problem with time windows.
Dynamic discretization discovery works on a time-expanded network, where only a small fraction of the vertices in the time-expanded networks is actually generated.
For the case of arc costs which are non-decreasing in time, \cite{VHBS2019} provide criteria allowing to prove optimality of a solution in the partially expanded network.
\cite{HBNS2019eprint} also apply dynamic discretization discovery, but for the time-dependent shortest path problem.
For the case of piecewise linear travel times satisfying the FIFO property, they propose an improvement of the dynamic discretization discovery based on an observation in \cite{Foschini2012} allowing to
explore only a small fraction of breakpoints of the arc travel time functions.

A survey on vehicle routing problems with time-dependent travel times given by \cite{GGG2015}.
\cite{DRvWdK2013} present an exact approach for the time-dependent vehicle routing problem.
\cite{HYI2008} propose a method to determine an optimal time schedule for a route where travel times and costs are piecewise linear and lower semi-continuous time-dependent functions.
\cite{VS2017} show how ready time functions can be determined efficiently by composition of piecewise linear travel time functions.

The research concerning time-dependent activity duration and the consumption of limited resources which may have to be replenished is scarce.
\cite{DB2008} analyse a project scheduling problem with time-dependent activity requirements and workers, which obey French workforce regulations and have various skill sets.
A mixed-integer programing formulation and heuristic solution approaches are presented.
For home health care, \cite{RH2016} present a scheduling approach considering break requirements for healthcare personnel assuming that staff members use a combination of walking and public transport. Travel times with public transport are modelled as time-dependent functions, where staff members may have to wait for the next tram, train, or metro.
The scheduling of breaks is conducted based on heuristic rules.
\cite{KHS2011} present a mixed-integer program for the problem of finding a schedule of minimal duration for a truck driver who must take breaks as required by European Union regulations.
The problem is extended by time-dependent travel times and used as a subproblem for route evaluations within an insertion heuristic. The subproblem is solved with a general purpose mixed-integer programming solver.
\cite{KLeff2019} studies the problem of finding a time-dependent shortest path from an origin to a destination in which breaks can be taken at parking lots in the network.
The driving time is limited by European Union regulations.
An approach combining forward exploration from the origin to a parking lot and backwards exploration from the destination to a parking is presented.

\section{Problem description}
Let us consider a sequence of activities $(1,2,\ldots,n)$ to be conducted in the given order and without preemption.
The time-dependent duration and resource consumption of each activity $i$ are given by functions $\tau_{i}(t_i)$ and $\rho_{i}(t_i)$, respectively, where $t_i$ refers to the start time of the activity.
Each activity~$i$ must start within a given time interval $[e_i,l_i]$ and the cumulative resource consumption must not exceed a given quantity $Q$.
With $\theta_{i}(t_i) = t_i + \tau_{i}(t_i)$ for all $i\in\{1,...,n\}$, the \emph{time-dependent activity scheduling problem (TDASP)} can be represented by

minimise
\begin{equation}\label{constr:base:completion}
\theta_n(t_n)
\end{equation}
subject to
\begin{equation}\label{constr:base:timing}
\theta_i(t_i) \leq t_{i+1}  \forall 1\leq i < n
\end{equation}

\begin{equation}\label{constr:base:capacity}
\sum_{i=1}^n \rho_i(t_i) \leq Q
\end{equation}
\begin{equation}\label{constr:base:t}
t_{i} \in[e_i,l_i]  \forall 1 \leq i  \leq n
\end{equation}

The objective (\ref{constr:base:completion}) is to minimise the completion time.
Constraint (\ref{constr:base:timing}) demands that the start time of any activity must not be before the completion time of the previous activity.
Constraint (\ref{constr:base:capacity}) requires that the cumulative resource consumption does not exceed the available amount.
Lastly, the domain of the variables is given by (\ref{constr:base:t}).

Throughout this paper we assume that the completion time functions $\theta_i(\cdot)$ are non-decreasing for all $i\in\{1,...,n\}$, meaning that completion times satisfy the FIFO property.
Without loss of generality we assume in the remainder of this paper that time windows are restricted in such a way that $\theta_i(e_i) \leq e_{i+1}$  and $\theta_i(l_i) \leq l_{i+1}$ for all $1\leq i < n$.

\begin{lemma}
Assume we have a sequence $t_1, t_2, \ldots, t_n$ with $t_1=e_1$ and $t_i = \max\{\theta_{i-1}(t_{i-1}),e_i\}$ for each $1< i\leq n$.
If $\rho_{i}(\cdot)$ is non-decreasing for $1\leq i\leq n$ and constraints~(\ref{constr:base:timing}) to~(\ref{constr:base:t}) are satisfied, then $t_1, t_2, \ldots, t_n$ is an optimal solution of the TDASP.
If any of the constraints~(\ref{constr:base:timing}) to~(\ref{constr:base:t}) is violated no feasible solution exists.
\end{lemma}
\begin{proof}
As the completion time functions $\theta_{i}(\cdot)$ are non-decreasing by the FIFO property,
the proof is trivial for non-decreasing resource consumption functions $\rho_{i}(\cdot)$.
\end{proof}

Resource consumption functions can be non-decreasing in scheduling problems with deterioration \citep[see, e.g.,][]{Mos1994}.
However, in other cases, the resource consumption may, e.g., depend on congestion levels and vary throughout the day.
In such cases it may be possible to reduce the total resource consumption by delaying the start time of some of the activities.
In the following sections we present solution approaches for cases in which resource consumption functions are not monotonous.

\subsection{Time-expanded networks}

An approximation of the solution of the TDASP can be obtained by discretising time and modelling the problem as a time-expanded network.
Let us assume we are given a parameter $\varepsilon>0$ for which $e_i$ and $l_i$ are a multiple of $\varepsilon$ for each activity~$i$.
Then, we can create for each activity $i$ and each time point $t_i \in \{e_i, e_i+\varepsilon, \ldots, l_i\}$
a vertex $(i,t_i)$ in the time-expanded network.
The edges of the time-expanded network are obtained by connecting every vertex $(i,t_i)$ with vertex $(i+1,t_{i+1})$ if and only if $t_{i+1} \geq\theta_i(t_i)$.
Figure~\ref{fig:time-expanded} illustrates such a time-expanded network with vertex sets $V_i$ for each~$1\leq i \leq n$.

\begin{figure}[htb]
\begin{center}
\mbox{%
\begin{tikzpicture}[->,>=stealth',shorten >=1pt,auto,node distance=1.5cm,thick]
\tikzset{label/.style={font=\large\bfseries}}
\tikzset{activity/.style={circle,draw,minimum size=10pt,font=\large\bfseries}}
\tikzset{replenishment/.style={circle,fill=black,draw,radius=3pt}}
  \node[activity] (1a) [] {};
  \node[activity] (1b) [below of=1a] {};
  \node[activity] (1c) [below of=1b] {};
  \node[label] (1dot) [below of=1c] {$\vdots$};

  \node[activity] (2a) [right = 2cm of 1a] {};
  \node[activity] (2b) [below of=2a] {};
  \node[activity] (2c) [below of=2b] {};
  \node[label] (2dot) [below of=2c] {$\vdots$};

  \node[activity] (3a) [right = 2cm of 2a] {};
  \node[activity] (3b) [below of=3a] {};
  \node[activity] (3c) [below of=3b] {};
  \node[label] (3dot) [below of=3c] {$\vdots$};

  \node[label] (4a) [right = 2cm of 3a] {$\hdots$};
  \node[label] (4b) [below of=4a] {$\hdots$};
  \node[label] (4c) [below of=4b] {$\hdots$};
  \node[label] (4dot) [below of=4c] {$\vdots$};

  \node[activity] (5a) [right = 2cm of 4a] {};
  \node[activity] (5b) [below of=5a] {};
  \node[activity] (5c) [below of=5b] {};
  \node[label] (5dot) [below of=5c] {$\vdots$};

  \node[label] (N1) [above = 0.25cm of 1a] {$V_1$};
  \node[label] (N2) [above = 0.25cm of 2a] {$V_2$};
  \node[label] (N3) [above = 0.25cm of 3a] {$V_3$};
  \node[label] (N4) [above = 0.25cm of 5a] {$V_n$};

  \node[label] (t1) [left = 0.25cm of 1a,text width=1.1cm] {$e_i$};
  \node[label] (t2) [left = 0.25cm of 1b,text width=1.1cm] {$e_i+\varepsilon$};
  \node[label] (t3) [left = 0.25cm of 1c,text width=1.1cm] {$e_i+2\varepsilon$};



  \path
    (1a) edge (2a)
    (1a) edge (2b)
    (1a) edge (2c)
    (1a) edge (2dot)
    (1b) edge (2c)
    (1b) edge (2dot)
    (1c) edge (2dot)
;

  \path
    (2a) edge (3a)
    (2a) edge (3b)
    (2a) edge (3c)
    (2a) edge (3dot)
    (2b) edge (3c)
    (2b) edge (3dot)
    (2c) edge (3dot)
;

  \path
    (3a) edge (4a)
    (3a) edge (4b)
    (3a) edge (4c)
    (3a) edge (4dot)
    (3b) edge (4c)
    (3b) edge (4dot)
    (3c) edge (4dot)
;

\path
  (4a) edge (5a)
  (4a) edge (5b)
  (4a) edge (5c)
  (4a) edge (5dot)
  (4b) edge (5c)
  (4b) edge (5dot)
  (4c) edge (5dot)
;

\end{tikzpicture}
}
\end{center}
\caption{A time-expanded network}\label{fig:time-expanded}
\end{figure}
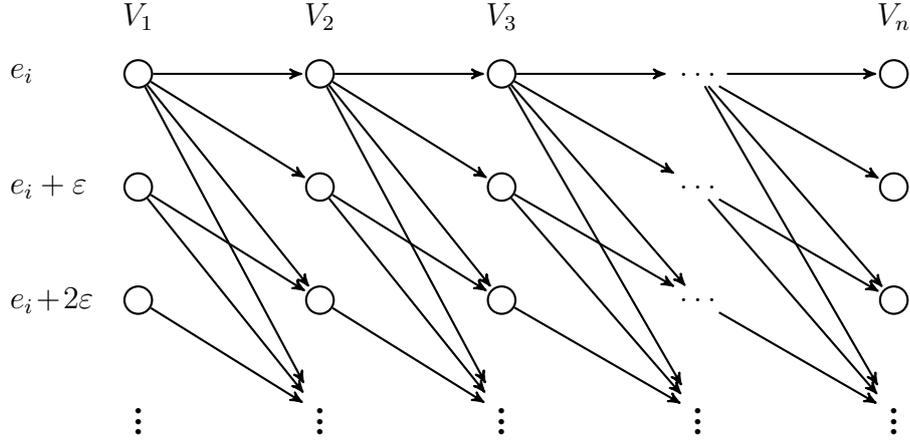

When setting $q_{(i,t)} = \rho_i(t)$ for all $(i,t) \in V = \bigcup_{i=1}^n V_i$, we can solve the TDASP by calling the function ${\tt solve}\big(V,\left(q_{(i,t)}\right)_{(i,t)\in V}\big)$ shown in Algorithm~\ref{alg:A*}.
Note, that we could have simply used $\rho_i(t)$ instead of $q_{(i,t)}$ in the algorithm, however, the used notation allows us in the following sections to use different values of $q_{(i,t)}$ to accelerate the search.
\begin{algorithm}[htbp]
\KwData{$(e_i,l_i,\tau_i,\theta_i,\rho_i)_{i\in\{1,\ldots,n\}}, Q$}
\Fn(){${\tt TEN\!::\!solve}\big(V,\left(q_{(i,t)}\right)_{(i,t)\in V}\big)$}{
$\bar V \leftarrow \emptyset$\;
 \ForEach{$(i,t) \in V$}{
   $\ell_{(i,t)} \leftarrow \left\{\begin{array}{ll}0 & \textrm{ if } i=1\\ \infty & \textrm{ else }\end{array}\right.$\;
 }
 \While{ exists $(i,t)\in V \setminus \bar V$ with $\ell_{(i,t)} + q_{(i,t)} \leq Q$}{
  $(i,t) \leftarrow \displaystyle\argminlex_{(i,t)\in V \setminus \bar V : \ell_{(i,t)} + q_{(i,t)} \leq Q} \Big\{ \big(\theta_i(t) + \sum_{j=i+1}^n \min_{t'\in [e_j,l_j]}\tau_j(t'),\ i, t \big) \Big\}$\;\label{line:selection}
  $\bar V \leftarrow \bar V \cup \{  (i,t) \}$\;
  \lIf { $i = n$ } {\textbf{break} \label{line:A*:break} }
  \ForEach{$t' \in \{t'  \mid (i+1,t') \in V\setminus \bar V,  \theta_i(t) \leq t'  \}$}{
   \If{ $\ell_{(i,t)} + q_{(i,t)} < \ell_{(i+1,t')}$  } {
    $\ell_{(i+1,t')} \leftarrow \ell_{(i,t)} + q_{(i,t)}$\;
    $p_{(i+1,t')} \leftarrow (i,t)$
   }\label{line:labelupdate}
  }
  }
 \Return ${\tt solution}\big(\bar V,\left(q_v,\ell_{v},p_v\right)_{v\in \bar V}\big)$\;
}
 \caption{Solving the TDASP in a time-expanded network}\label{alg:A*}
\end{algorithm}

The algorithm  is similar to an A*-Algorithm for shortest path problems \citep{Astar}.
However, labels are not related to objective function values.
Instead, labels indicate the cumulative resource consumption until the start of each activity at a particular time.
The algorithm  starts with initializing a set of vertices $\bar V$ for which the label is permanently set.
It tentatively sets all labels associated to the first activity to zero, and all other labels to infinity.
As long as there is a vertex in $V\setminus \bar V$ which has a label allowing to conduct the activity without exceeding the capacity, the algorithm selects such a vertex~$(i,t)\in V\setminus \bar V$ according to a lexicographical ordering on the lower bound on the completion time of the last activity, the activity index, and the time associated to the vertex.
The lexicographical ordering is necessary to ensure that all relevant predecessors are considered before selecting a vertex.
The selected vertex is inserted in $\bar V$.
If the selected vertex belongs to the last activity, the algorithm terminates.
Otherwise, non-permanent labels of the next activity are updated if their value can be reduced.
Whenever a label is updated, the predecessor $p_v$ for the vertex is updated so that the solution can be reconstructed by a function  ${\tt solution}\big(\bar V,\left(q_v,\ell_{v},p_v\right)_{v\in \bar V}\big)$.
If a solution is found, this function returns the vertices of the solution obtained by selecting the earliest vertex $(n,t)$ with $\ell_{(n,t)} + q_{(n,t)} \leq Q$ and iterating through its predecessors.
If no solution is found, the function returns $\big\{ (1,\infty), \ldots, (n,\infty) \big\}$.
For brevity reasons we omit a detailed description of this trivial function.
\begin{lemma}
\label{lemma:A*}
Algorithm~\ref{alg:A*} finds an optimal solution of the TDASP restricted to the time-expanded network with vertex set $V$ if and only if such a solution exists.
\end{lemma}

\begin{proof}
Assume that vertex $(i,t)$ is selected in an iteration of the algorithm.
Due to the FIFO property, each potential predecessor of $(i,t)$ must have been selected in an earlier iteration.
Thus, $\ell_{(i,t)}$ represents the minimal cumulative resource consumption before conducting activity~$i$ at time $t$.
If $(n,t)$ is selected and $\ell_{(n,t)} + q_{(n,t)} \leq Q$, then $t$ is the earliest time for which $\ell_{(n,t)} + q_{(n,t)} \leq Q$.
Thus, the solution obtained by backtracking the predecessors is optimal.
If the algorithm terminates without having selected a vertex $(n,t)$ with $\ell_{(n,t)} + q_{(n,t)} \leq Q$, then there is no path $(1,t_1), \ldots, (n,t_n)$  with $(i,t_i)\in V$ for $i\in\range{1}{n}$,
$\theta_i(t_i) \leq t_{i+1}$ for $i\in\range{1}{n-1}$,
and  $\sum_{i=1}^n q_{(i,t_i)} \leq Q$.
\end{proof}

\subsection{Dynamic Discretization Discovery}

The difficulty in solving the TDASP with discretized times is that the number of vertices in the time-expanded network can become very large for small values of $\varepsilon$, i.e. the parameter influencing the granularity of the time discretization.
In order to reduce the number of vertices we now show how to solve the TDASP with discretized times on a partially expanded network similar to the dynamic discretization discovery approach by \cite{VHBS2019} for the time-dependent traveling salesman problem with time windows.
Unlike \cite{VHBS2019}, who assume that arc costs are non-decreasing in time, we assume that resource consumption functions can be non-monotonous.

In the remainder of this section we will refer to a partially expanded network if the set of vertices $V$ is a subset of the vertices of the fully expanded network and
the set of edges is obtained by connecting every vertex $(i,t_i)$ with vertex $(i+1,t_{i+1})$ if and only if $t_{i+1} \geq\theta_i(t_i)$.
Our approach will generate and modify such partially expanded networks with the following properties.

\begin{property}
\label{prop:earliest-latest}
For each $1 \leq i \leq n$ we have $(i,e_i) \in V$ and $(i,l_i) \in V$.
\end{property}

\begin{property}
\label{prop:successor}
For each vertex $(i,t) \in V$ with $e_{i+1} < \varepsilon \left\lceil \frac{\theta_i(t)}{\varepsilon} \right\rceil < l_{i+1}$, we have $\Big(i+1, \varepsilon \left\lceil \frac{\theta_i(t)}{\varepsilon} \right\rceil \Big)\in V$.
\end{property}

\begin{property}
\label{prop:consumption}
Each vertex $(i,t)\in V$ with $t<l_i$ is assigned a consumption value of
\[q_{(i,t)} = \min_{\bar t \in \{ t, t + \varepsilon, \ldots , t' -\varepsilon\}} \rho_i(\bar t)\]
 where $t'>t$ is the smallest value for which $(i,t')\in V$.
Each vertex $(i,l_i)\in V$ is assigned a consumption value of $q_{(i,l_i)} = \rho_i(l_i)$.
\end{property}

Property~\ref{prop:earliest-latest} guarantees that for each activity $i$, a vertex representing the earliest and latest start time is included in the partially expanded network.
For every vertex, its immediate successor is included in the partially expanded network if Property~\ref{prop:successor} is satisfied.
Because of Property~\ref{prop:consumption}, each vertex $(i,t)$ is assigned a value $q_{(i,t)}$, which is a lower bound on the actual resource consumption of the activity for any start time between $t$ and $t'-\varepsilon$, where $t'$ refers to the next vertex of activity~$i$.

With these properties we can make the following observation.

\begin{lemma}
\label{lowerbound}
The optimal solution of a TDASP restricted to a partially expanded network satisfying Properties~\ref{prop:earliest-latest} to~\ref{prop:consumption} is a lower bound on the optimal solution of the fully expanded TDASP with discretized times.
\end{lemma}
\begin{proof}
Assume that vertices $(1,t_1),$ $\ldots, (n,t_n)$ correspond to an optimal solution of the fully expanded TDASP with discretized times.
Because of Properties~\ref{prop:earliest-latest} and~\ref{prop:successor}, we can find a path $(1,t'_1),$ $\ldots, (n,t'_n)$ in the partially expanded network, where  for each $1\leq i\leq n$ the value $t'_i$ is the largest while not exceeding $t_i$ such that $(i,t'_i)$ is in the partially expanded network.
Because of Property~\ref{prop:consumption}, we have $q_{(i,t'_i)} \leq \rho_i(t_i)$ for each $1\leq i\leq n$
and $\sum_{i=1}^n q_{(i,t'_i)} \leq \sum_{i=1}^n \rho_i(t_i) \leq Q$.
Therefore, $(1,t'_1),$ $\ldots, (n,t'_n)$ is feasible for the TDASP restricted to the partially expanded network and has a completion time that is smaller or equal to the solution  of the fully expanded TDASP with discretized times.
\end{proof}

\begin{lemma}
\label{optimality}
Assume that the path $(1,t_1),$ $\ldots, (n,t_n)$ corresponds to an optimal solution of the TDASP restricted to a partially expanded network satisfying Properties~\ref{prop:earliest-latest} to~\ref{prop:consumption}.
Furthermore, assume that $q_{(i,t_i)} = \rho_i(t_i)$ for each $1\leq i\leq n$.
Then, the path corresponds to an optimal solution of the fully expanded TDASP with discretized times.
\end{lemma}
\begin{proof}
If  $q_{(i,t_i)} = \rho_i(t_i)$ for each $1\leq i\leq n$, then $\sum_{i=1}^n \rho_i(t_i) = \sum_{i=1}^n q_{(i,t_i)} \leq Q$.
Therefore, $(1,t_1),$ $\ldots, (n,t_n)$ is feasible for the TDASP with discretized times. Because of Lemma~\ref{lowerbound} we can conclude that the path corresponds to an optimal solution of the TDASP with discretized times.
\end{proof}

\begin{algorithm}[htbp]
\KwData{$(e_i,l_i,\tau_i,\theta_i,\rho_i)_{i\in\{1,\ldots,n\}}, Q, \varepsilon$}
\Fn(){${\tt DDD\!::\!initialise}()$}{
 $V \leftarrow \emptyset$\;
 \For{$i\leftarrow n$ \KwTo $1$}{
$\big( V, (q_v)_{v\in V} \big) \leftarrow {\tt DDD\!::\!addRecursive}\big(V, (q_v)_{v\in V},(i,l_i)\big)$\;
  $\big( V, (q_v)_{v\in V} \big) \leftarrow {\tt DDD\!::\!addRecursive}\big(V, (q_v)_{v\in V},(i,e_i)\big)$\;
 }
 \Return $\big( V, (q_v)_{v\in V} \big)$\;
}
\Fn(){${\tt DDD\!::\!addRecursive}\big(V, (q_v)_{v\in V},(i,t)\big)$}{
 \lIf{$(i,t)\in V$}{ \Return $\big( V, (q_v)_{v\in V} \big)$ }
 $V \leftarrow V \cup \{ (i,t) \}$\;
 \uIf{$t = l_i$}{
    $q_{(i,t)} \leftarrow \rho_i(t)$}
 \Else{
  $t' \leftarrow \min\{ t'  \mid (i,t') \in V, t' > t\}$\;
  $q_{(i,t)} \leftarrow \displaystyle\min_{\bar t \in \{ t, t+\varepsilon, \ldots , t' -\varepsilon\}} \rho_i(\bar t)$\;
  \If{$t > e_i$}{
    $t' \leftarrow \max\{ t ' \mid (i,t') \in V, t' < t\}$\;
    $q_{(i,t')} \leftarrow \displaystyle\min_{\bar t \in \{ t', t'+\varepsilon, \ldots , t -\varepsilon\}} \rho_i(\bar t)$\;
  }
 }
 \uIf{$i = n$\label{line:recurse}}{ \Return $\big( V, (q_v)_{v\in V} \big)$ }
 \Else{ \Return ${\tt DDD\!::\!addRecursive}\Big(V, (q_v)_{v\in V},\big(i+1,\max\{e_{i+1},\varepsilon\big\lceil\frac{\theta_i(t)}{\varepsilon}\big\rceil\}\big)\Big)$ }
}

 \caption{Initialising a partial time-expanded network}\label{alg:init}
\end{algorithm}

We can use Algorithm~\ref{alg:init} to create an initial partially expanded network.
According to Property~\ref{prop:earliest-latest}, the algorithm adds vertices corresponding to each activity at the start and end of the time window.
These vertices are added in the order  $(n,l_n), (n,e_n), \ldots, (1,l_1), (1,e_1)$  and their immediate successors are recursively added according to Property~\ref{prop:successor}.
When adding a vertex $(i,t)$ the consumption values of the new vertex and the preceding vertex $(i,t')$ are set as required by Property~\ref{prop:consumption}.

\begin{algorithm}[htb]
\KwData{$(e_i,l_i,\tau_i,\theta_i,\rho_i)_{i\in\{1,\ldots,n\}}, Q, \varepsilon$}
\Fn(){${\tt DDD\!::\!solve}()$}{
 $\big( V, (q_v)_{v\in V} \big) \leftarrow {\tt DDD\!::\!initialise}()$\;
 \Repeat{$\rho_i(t_i) = q_{(i,t_i)}  \forall 1 \leq i \leq n$}{
  $\big\{ (1,t_1), \ldots, (n,t_n) \big\} = {\tt TEN\!::\!solve}\big( V, (q_v)_{v\in V} \big)$\;
  \lIf{ $t_n=\infty$}{\textbf{break}}
  \If{ exists $1  \leq i \leq n$ with $\rho_i(t_i) > q_{(i,t_i)} $} {
   select $(i,t_i)$ with $\rho_i(t_i) > q_{(i,t_i)}$\;
   $t \leftarrow \min\{ t' \mid t' > t_i, (i,t') \in V\}$\;
   \Repeat{$\displaystyle\min_{\bar t \in \{ t_i, t_i+\varepsilon, \ldots , t - \varepsilon\}} \rho_i(\bar t)  > q_{(i,t_i)}$} {
    $t \leftarrow \varepsilon \lceil\frac{t_i+t}{2 \varepsilon}\rceil$\;
   }
   $\big(V, (q_v)_{v\in V}\big) \leftarrow  {\tt DDD\!::\!addRecursive}\big(V, (q_v)_{v\in V}, (i,t)\big)$\;
  }
 }
 \Return $\big\{ (1,t_1), \ldots, (n,t_n) \big\}$\;
}

 \caption{Dynamic discretization discovery algorithm}\label{alg:DDD}
\end{algorithm}

Algorithm~\ref{alg:DDD} solves the TDASP with discretized times by dynamically creating a time-expanded network.
The algorithm starts by initialising the network using Algorithm~\ref{alg:init}.
Then it iterates until a solution of the TDASP with discretized times is found.
Each iteration begins with finding a solution of the TDASP restricted to the partially expanded network using Algorithm~\ref{alg:A*}.
If there is an activity~$i$ in the solution that begins at a time $t$ such that $\rho_i(t) > q_{(i,t)}$, then the total consumption associated with the solution may exceed the capacity.
Therefore, a new vertex is added to the partially expanded network between $(i,t)$ and the next vertex of activity~$i$.
The time for the new vertex is selected in such a way that the value of $q_{(i,t_i)}$ is increased when adding the new vertex.

If a feasible solution exists, the algorithm either terminates with the optimal solution, or adds at least one vertex to the partially expanded network in each iteration.
Therefore, the algorithm is guaranteed to terminate with the optimal solution.
In the worst case, the algorithm may have to fully expand the network.
However, in general, the algorithm can be expected to require only a fraction of the vertices.

\section{Replenishments}
So far we assumed that the resource cannot be used after it has been consumed by the activities.
This is the case, for example, if electric batteries cannot be recharged or if maximum working hours or operating times are reached.
In many cases, however the resource can be regenerated.
For example, electric batteries can be recharged, workers are allowed to work again after taking a sufficiently long break, or machines can be used again after some maintenance has been conducted.
In this section we consider an extension of the TDASP in which the resource can be replenished between the end of an activity and the begin of the next activity.
We assume that the time required to fully replenish a resource after conducting activity $i$ can be determined by a consumption-dependent and non-decreasing replenishment function $\Delta_i(Q_i)$, where $Q_i$ is the cumulative resource consumption after completing activity~$i$.

The resulting \emph{time-dependent activity scheduling problem with replenishments (TDASPR)} is to

minimise
\begin{equation}\label{constr:completion:replenishments}
\theta_n(t_n)
\end{equation}
subject to
\begin{equation}\label{constr:timing:replenishments}
\theta_i(t_i) +  y_{i} \Delta_{i}(Q_i)  \leq t_{i+1}  \forall 1\leq i < n
\end{equation}
\begin{equation}\label{constr:consumption_1:replenishments}
Q_1 = \rho_1(t_1)
\end{equation}
\begin{equation}\label{constr:consumption_i:replenishments}
Q_{i+1} = (1-y_i)Q_i + \rho_{i+1}(t_{i+1})\forall 1 \leq i < n
\end{equation}
\begin{equation}\label{constr:domains:replenishments}
t_{i} \in[e_i,l_i],
Q_i \in [0,Q],
y_{i} \in \{0,1\}  \forall 1 \leq i  \leq n.
\end{equation}
In this problem $Q_i$ is a variable representing the cumulative resource consumption between the last replenishment and completion of activity~$i$, and $y_i$ is a binary variable representing whether the resource is replenished after conducting activity~$i$.
The objective function is the same as in the case without replenishments.
Constraint~(\ref{constr:timing:replenishments}) is analogous to Constraint~(\ref{constr:base:timing}) but includes the replenishment duration if necessary.
Constraints~(\ref{constr:consumption_1:replenishments}) and~(\ref{constr:consumption_i:replenishments}) ensure that the cumulative resource consumption is correctly determined
and
Constraints~(\ref{constr:domains:replenishments}) restricts the domains of the variables.

The special case where replenishments must or must not be taken after some of the activities is covered in Appendix~\ref{app:restrictedreplenishments}.
\subsection{Replenishments in time-expanded networks}

In the following we show how dynamic discretization discovery can be used to solve the TDASPR.
For this, we replace Algorithm~\ref{alg:A*} by Algorithm~\ref{alg:replenishments:A*}
which dynamically adds new vertices corresponding to a fully replenished resource to the partially expanded network.

\begin{algorithm}[p]
\KwData{$(e_i,l_i,\tau_i,\theta_i,\rho_i,\Delta_i)_{i\in\{1,\ldots,n\}}, Q, \varepsilon$}
\Fn(){${\tt TEN}{\hbox{-}}\Delta{\tt ::\!solve}\big(V,\left(q_{(i,t)}\right)_{(i,t)\in V}\big)$}{
$\bar V \leftarrow \emptyset$\;
 \ForEach{$(i,t) \in V$}{
   $\ell_{(i,t)} \leftarrow \left\{\begin{array}{ll}0 & \textrm{ if } i=1\\ \infty & \textrm{ else }\end{array}\right.$\;
 }
 \While{ exists $(i,t)\in V \setminus \bar V$ with $\ell_{(i,t)} + q_{(i,t)} \leq Q$}{
  $(i,t) \leftarrow \displaystyle\argminlex_{(i,t)\in V \setminus \bar V : \ell_{(i,t)} + q_{(i,t)} \leq Q} \Big\{ \big(\theta_i(t) + \sum_{j=i+1}^n \min_{t'\in [e_j,l_j]}\tau_j(t'),\ i, t \big) \Big\}$\;\label{line:selection}
  $\bar V \leftarrow \bar V \cup \{  (i,t) \}$\;
  \lIf { $i = n$ } {\textbf{break}\label{line:replenishments:A*:break} }
  $t^* \leftarrow \varepsilon\big\lceil\frac{\theta_i(t) + \Delta_i(\ell_{(i,t)}+ q_{(i,t)})}{\varepsilon}\big\rceil$\;\label{line:earliestreplenishmenttime}
  \If{ $t^* \leq l_{i+1}$ 
  }{
     ${\tt DDD\!::\!addRecursive}\big(V, (q_v)_{v\in V},(i+1,t^*))$\label{line:earliestreplenishment}
  }
  \ForEach{$t' \in \{t'  \mid (i+1,t') \in V\setminus \bar V,  \theta_i(t) \leq t' < t^* \}$\label{line:noreplenishmentloop}}{
   \If{ $\ell_{(i,t)} + q_{(i,t)} < \ell_{(i+1,t')}$  } {
    $\ell_{(i+1,t')} \leftarrow \ell_{(i,t)} + q_{(i,t)}$\;
    $p_{(i+1,t')} \leftarrow (i,t)$
   }
  }\label{line:noreplenishmentloopend}
  \ForEach{$t' \in \{t'  \mid (i+1,t') \in V\setminus \bar V,  t' \geq t^* \}$\label{line:replenishmentloop}}{
   $\ell_{(i+1,t')} \leftarrow 0$\;
   $p_{(i+1,t')} \leftarrow (i,t)$
  }
  }
 \Return ${\tt solution}\big(\bar V,\left(q_v,\ell_{v},p_v\right)_{v\in \bar V}\big)$\;
}
 \caption{Solving the TDASPR in a partially expanded network}\label{alg:replenishments:A*}
\end{algorithm}

Whenever a vertex $(i,t)$ with $i<n$ is selected in Algorithm~\ref{alg:replenishments:A*},
the algorithm determines the earliest replenishment time $t^*$ after conducting activity~$i$ at time~$t$.
If $t^*\leq l_{i+1}$, the vertex $(i+1,t^*)$ is added to $V$ using algorithm ${\tt DDD\!::\!addRecursive}$.
Then, all labels for vertices with a time prior to $t^*$ are updated if their value can be reduced and
all labels  for vertices with a later time are set to zero.

\begin{lemma}
\label{replenishments:lowerbound}
Assume that Algorithm~\ref{alg:replenishments:A*} is applied to a partially expanded network satisfying Properties~\ref{prop:earliest-latest} to~\ref{prop:consumption}.
Then, the solution of the algorithm is a lower bound on the solution in the fully expanded network.
\end{lemma}
\begin{proof}
Let us first note that the solution returned by Algorithm~\ref{alg:replenishments:A*} is the same as the solution returned by Algorithm~\ref{alg:replenishments:A*} without Line~\ref{line:replenishments:A*:break}.
For the ease of argumentation, we therefore base our proof on a variant of Algorithm~\ref{alg:replenishments:A*} without Line~\ref{line:replenishments:A*:break}.
Let $(1,t_1),$ $\ldots, (n,t_n)$ denote an optimal solution in the fully expanded network.
Let $\bar V$ denote the set of permanently labelled vertices after execution of this variant of the algorithm, including the vertices that may have been added to $V$ during the course of the algorithm.
Let $(1,\bar t_1),$ $\ldots, (n,\bar t_n) \in \bar V$ denote the vertices returned by the algorithm.
Analogously to Lemma~\ref{lemma:A*} we can show that $\bar t_n$ is smaller or equal to
the completion time of any other path through vertices in $\bar V$ not exceeding the capacity between replenishments.

Let $(1,t'_1),$ $\ldots, (n,t_n')$ denote the vertices in $\bar V$ for which $t'_i$ is the largest time not exceeding~$t_i$ for each $i\in\range{1}{n}$.
These vertices can be found because of Property~\ref{prop:earliest-latest}.
As  $\theta_i$ is non-decreasing for each $i\in\range{1}{n}$, Property~\ref{prop:successor} guarantees that
$\theta_i(t'_i) \leq t'_{i+1}$ for $i\in\range{1}{n-1}$.

Assume, that in the optimal solution $(1,t_1),$ $\ldots, (n,t_n)$, the first replenishment is conducted just before starting activity $j$.
Because of Property~\ref{prop:consumption}, we have $\sum_{i=1}^{j-1} q_{(i,t'_i)} \leq \sum_{i=1}^{j-1} \rho_i(t_i) \leq Q$.
Because $\Delta_{j-1}$ is non-decreasing,
we have $\theta_{j-1}(t'_{j-1}) + \Delta_{j-1}\big(\sum_{i=1}^{j-1} q_{(i,t'_i)} \big) \leq \theta_{j-1}(t_{j-1}) + \Delta_{j-1}\big(\sum_{i=1}^{j-1} \rho_i(t_i) \big)$.
Therefore, there is a vertex $(j,t)\in \bar V$ with $t\leq t_j$ and $\ell_{(j,t)} = 0$.
By definition we have $t \leq t'_j \leq t_j$ and we can conclude that $\ell_{(j,t'_j)} = 0$.
Analogously we can show that $\sum_{i=j}^{k-1} q_{(i,t'_i)} \leq \sum_{i=j}^{k-1} \rho_i(t_i) \leq Q$ and $\ell_{(k,t'_k)} = 0$ if the next replenishment in the solution $(1,t_1),$ $\ldots, (n,t_n)$ is conducted just before starting activity $k$.
We can apply this argument repeatedly to conclude that the capacity is not exceeded along the path $(1,t'_1),$ $\ldots, (n,t'_n)$
and that $\ell_{(n,t'_n)} + q_{(n,t'_n)} \leq Q$.
Thus, $\bar t_n \leq t'_n \leq t_n$.
\end{proof}

\begin{lemma}
\label{replenishments:optimality}
Assume that the path $(1,t_1),$ $\ldots, (n,t_n)$ corresponds to a solution created by Algorithm~\ref{alg:DDD} using Algorithm~\ref{alg:replenishments:A*}.
Furthermore, assume that $q_{(i,t_i)} = \rho_i(t_i)$ for each $1\leq i\leq n$.
Then, the path corresponds to an optimal solution of the fully expanded TDASPR with discretized times.
\end{lemma}
\begin{proof}
For each $1\leq i\leq n$ let $\ell_{(i,t_i)}$ denote the corresponding label created by the algorithm.
Let~$i$ be any activity index with $\ell_{(i,t_i)}=0$.
For any activity index $j > i$ such that no replenishment is conducted between $i$ and $j$, we have $Q \geq \ell_{(j,t_j)} + q_{(j,t_j)}= \sum_{k=i}^{j} q_{(k,t_k)} = \sum_{k=i}^{j} \rho_k(t_k)$.
Therefore, the duration of any replenishment is sufficiently long
and the capacity is never exceeded between replenishments.
Thus, $(1,t_1),$ $\ldots, (n,t_n)$ is a feasible path in the fully expanded network. Because of Lemma~\ref{replenishments:lowerbound} we can conclude that the path corresponds to an optimal solution of the TDASP with discretized times.
\end{proof}

\section{Preloading of vertices}

Assume we have already solved the problem for a sequence of activities $(1, \ldots, n)$ and we want to solve the problem for a sequence $(1, \ldots, n, n+1, \ldots, m)$.
It is likely that some of the vertices that are generated when scheduling activities $(1, \ldots, n)$ are also generated when scheduling activities $(1, \ldots, n, n+1, \ldots, m)$.
We can reduce the computational effort associated to finding such vertices again by preloading them before starting the solution process for sequence $(1, \ldots, n, n+1, \ldots, m)$.
On the other hand, preloading all of the vertices generated when determining a schedule for activities $(1, \ldots, n)$ may create an unnecessary computational overhead, because only a share of the vertices can be expected to be required to find a solution for activities $(1, \ldots, n, n+1, \ldots, m)$.

To balance the tradeoff between salvaging previous computational effort and avoiding unnecessary overhead by adding too many vertices, we propose
to preload all vertices of the solution path for sequence $(1, \ldots, n)$ to our partially expanded network after initialization and before starting the solution process.
In order to ensure that $q_{(i,t)} = \rho_i(t)$ for all vertices preloaded, we also add the immediate successors.
As all vertices are added using the function ${\tt DDD\!::\!addRecursive}$, Properties~\ref{prop:earliest-latest} to~\ref{prop:consumption}  maintain satisfied.
Consequently, the optimality conditions for all vertices on the previous solution path are satisfied.

The benefit of adding these vertices after initialization of the partially expanded network is that the total number of vertices generated by the dynamic discretization discovery may be drastically reduced for those cases where the solution for sequence  $(1, \ldots, n, n+1, \ldots, m)$ overlaps with the solution for sequence  $(1, \ldots, n)$.
In some cases, however, preloading may cause a small overhead because vertices may be recursively added that are not needed for the solution process.

\section{Evaluation}

This paper is motivated by an industry project on last-mile deliveries with electric vehicles, where a homogeneous fleet of electric delivery vehicles is used to deliver parcels within metropolitan areas, i.e., regions consisting of densely populated city centres and less-populated surrounding territories.
Traffic conditions depend both on the particular area as well as the time of the day with morning and afternoon peak hours.
With increasing distance to the city centres, the impact of peak hours on travel times decreases.
Thus, the time-dependency of travel times is not homogeneous and delays caused by congestion can be avoided by circumventing city centres during peak hours.
Current electric vehicles have minimal consumption per unit distance at moderate speeds while the consumption increases towards very low or high velocities.
Therefore, the energy consumption of the vehicles is related to travel speeds, which in turn depend on the congestion level.
All vehicles leave the depot fully charged and have to return to the depot before running out of energy.
They can recharge the battery at the depot as well as at public charging infrastructure.

Due to confidentiality reasons, we are not allowed to evaluate our approaches on instances obtained from data of our project partners.
Therefore, we generate artificial instances for the time-dependent vehicle routing problem with replenishments having the characteristics described above.
Our instances are based on the well-known instances for the vehicle routing problem with time windows by \cite{Solomon1987}.
These instances have customer locations distributed in the Euclidean plane
and are subdivided into instance sets with clustered customer locations (C1 and C2), randomly distributed customer locations (R1 and R2), and mixed locations (RC1 and RC2).

\begin{figure}[htb]
      \centering
        \includegraphics[height=4cm]{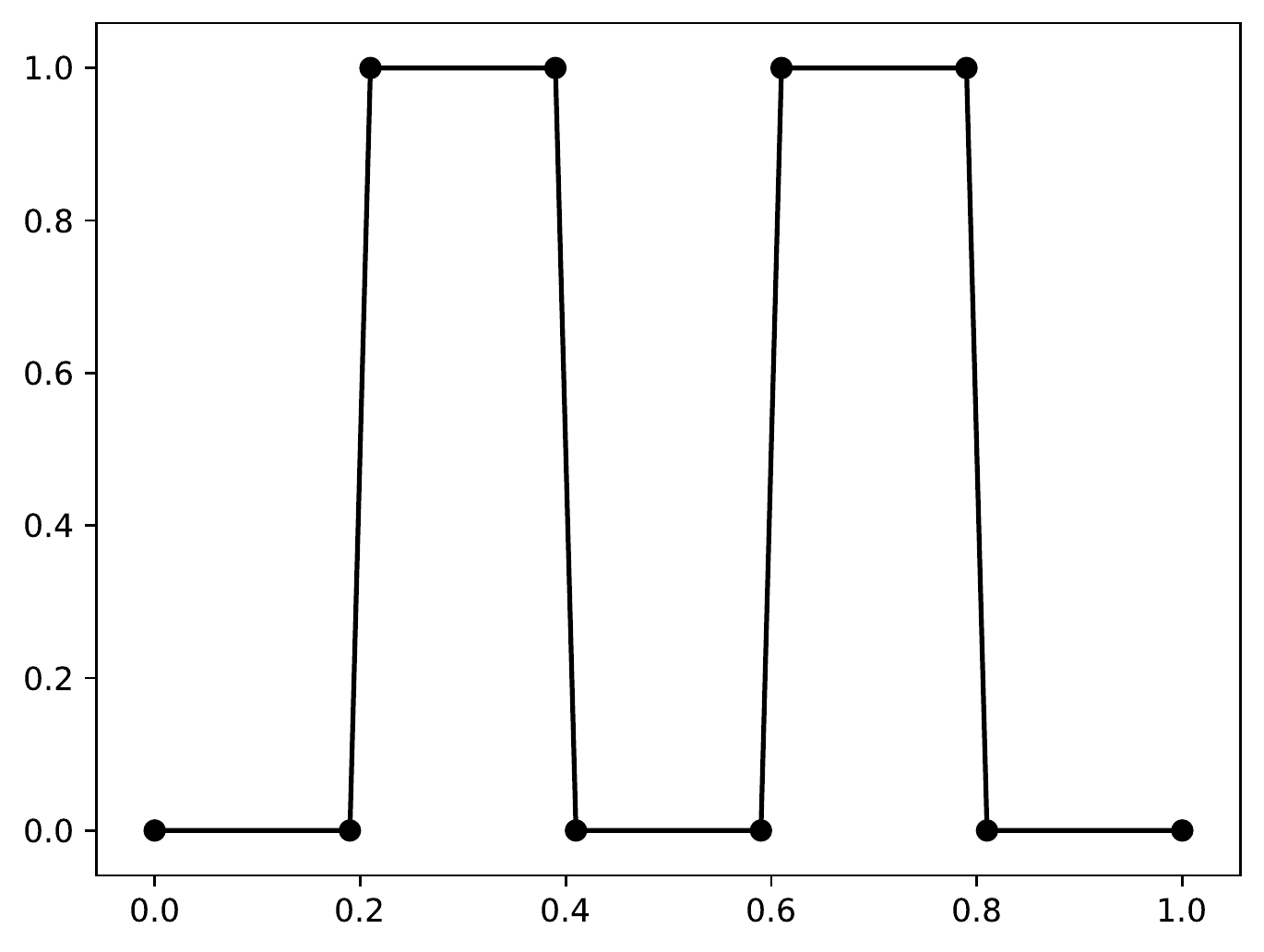}
\caption{Time-dependent congestion $\delta(t)$}\label{fig:congestionfactors}
\end{figure}
In order to include time-dependent travel times,  we introduce a time-dependent congestion factor $\delta(t) \in [0,1]$ for each point in time as shown in Figure~\ref{fig:congestionfactors}.
The time values on the horizontal axis in the figure are given relative to the duration of the time horizon which differs throughout the instances.
Low levels of $\delta(t)$ indicate free-flow traffic conditions during off-peak hours.
High values of $\delta(t)$ indicate high congestion levels during morning and afternoon peak-hours.

\begin{figure}[htbp]
      \centering
      \begin{subfigure}[b]{0.4\textwidth}
        \includegraphics[width=5.5cm]{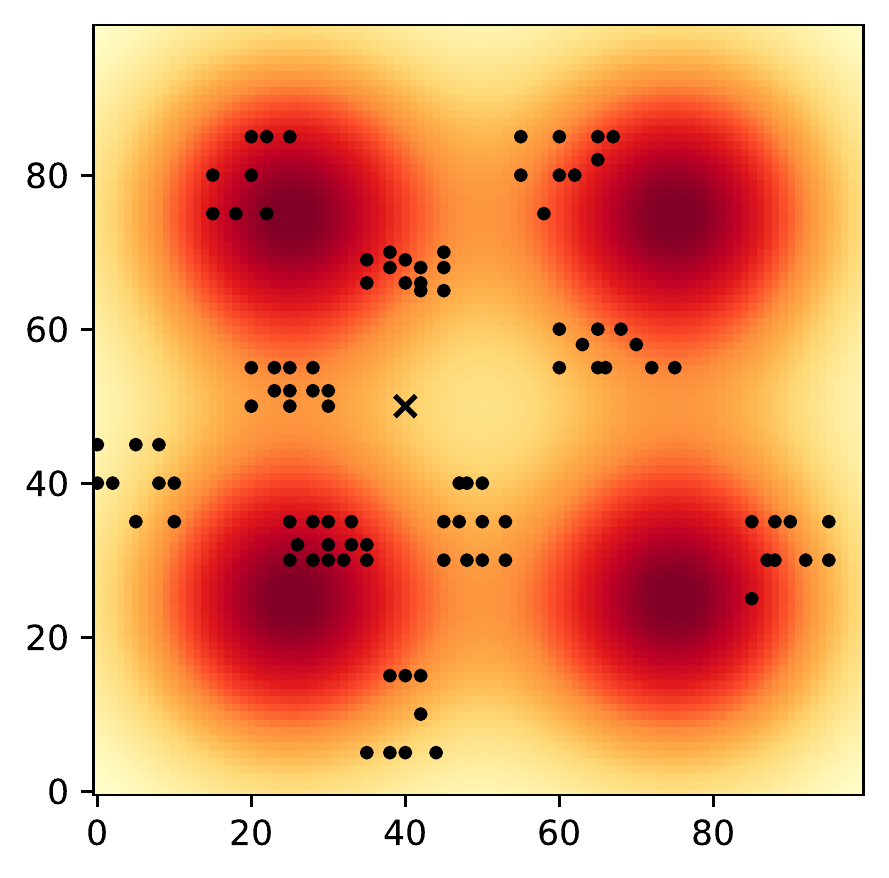}
        \subcaption{C1 instances}
      \end{subfigure}
      \
      \begin{subfigure}[b]{0.4\textwidth}
        \includegraphics[width=5.5cm]{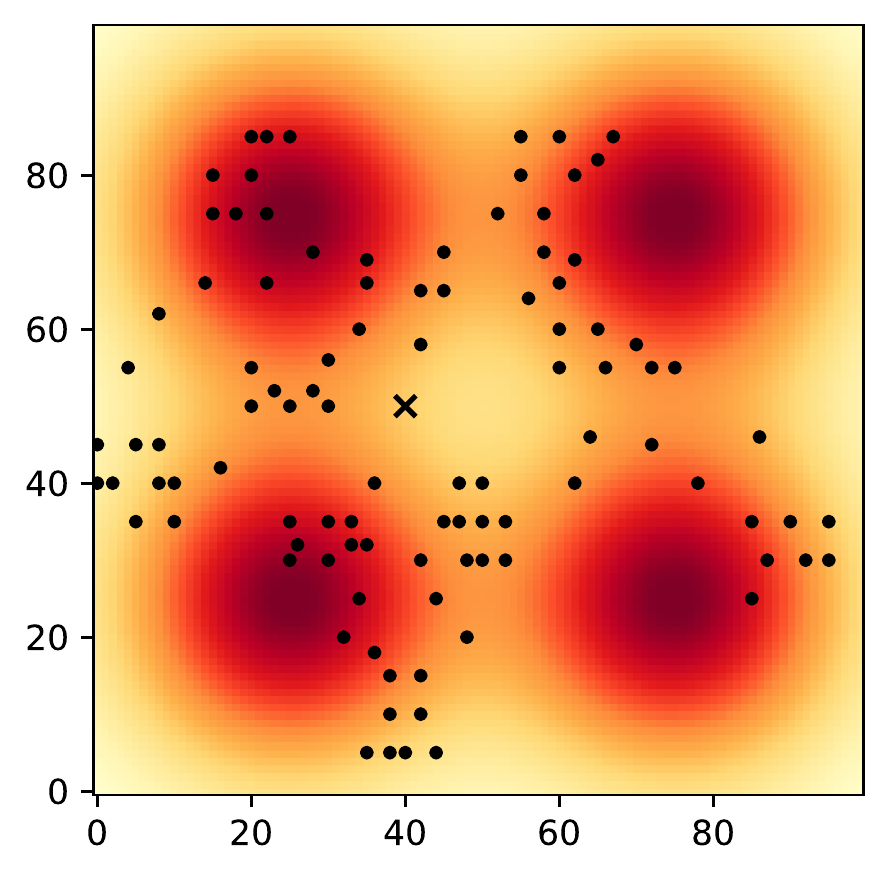}
        \subcaption{C2 instances}
      \end{subfigure}
      \\
      \begin{subfigure}[b]{0.4\textwidth}
        \includegraphics[width=5.5cm]{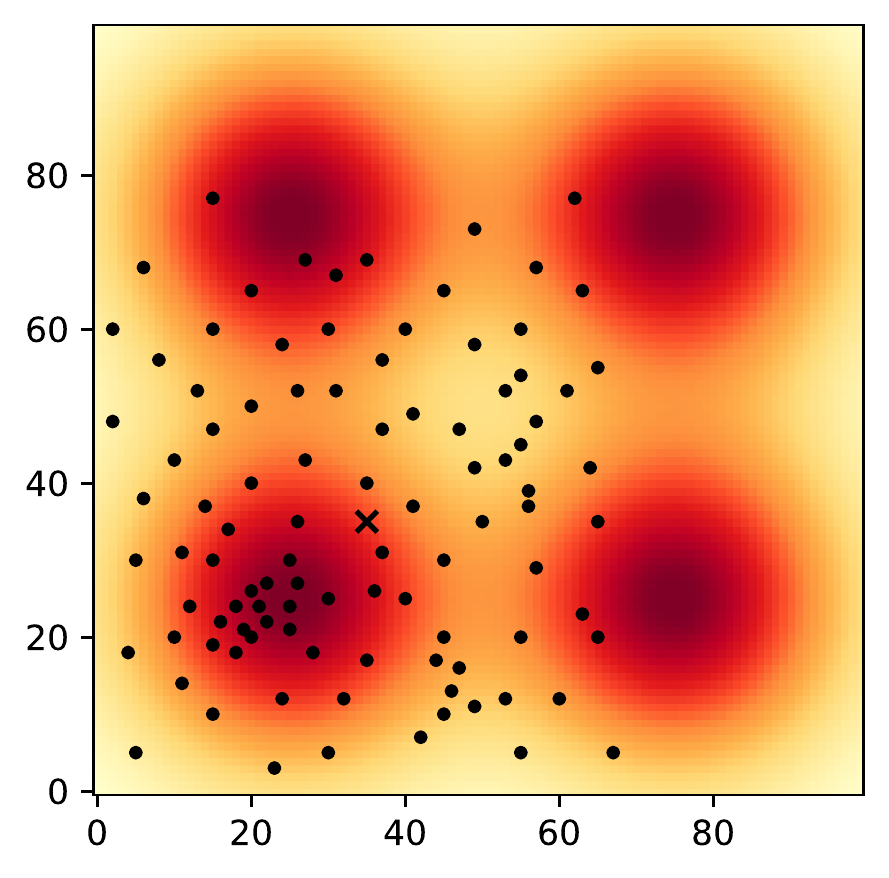}
        \subcaption{R1 \& R2 instances}
      \end{subfigure}
      \
      \begin{subfigure}[b]{0.4\textwidth}
        \includegraphics[width=5.5cm]{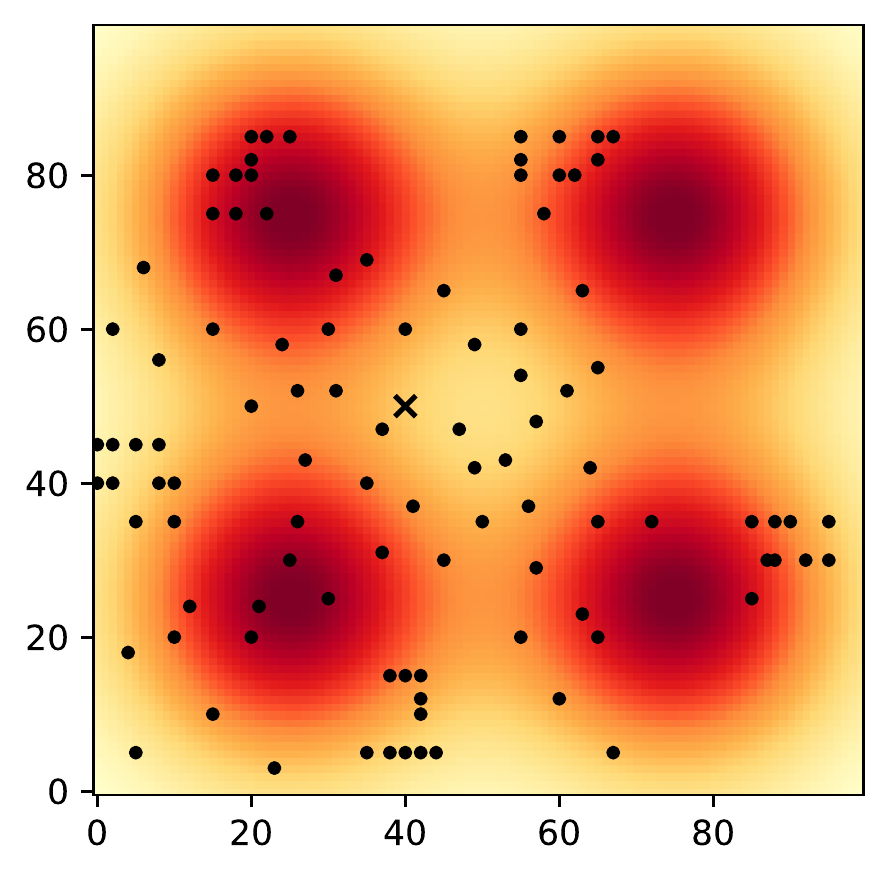}
        \subcaption{RC1 \& RC2 instances}
      \end{subfigure}
\caption{Location-dependent congestion $\gamma(x,y)$ and customer locations.
}\label{fig:customerlocations}
\end{figure}
We assume a city centre in each quadrant of the Euclidean plane where congestion is maximal.
In order to consider that the level of congestion decreases with increasing distance to the city centres,
 we calculate a location-dependent congestion factor $\gamma(x,y) \in [0,1]$ by the Gaussian function based on the distance to the city centres.
Figure~\ref{fig:customerlocations} shows the location-dependent congestion factors, the customer locations (shown as $\bullet$), and the location of the depot (shown as $\times$) for the different subsets of instances proposed by \cite{Solomon1987}.
Low values of $\gamma(x,y)$ indicate low sensitivity to congestion and are illustrated in light color.
High values of $\gamma(x,y)$ indicate high sensitivity to congestion and are illustrated in dark color.
We can see that all instances comprise customers located in areas prone to congestion and others in areas where congestion is negligible.
Even if two customers are located in areas without significant congestion, the direct path from one to the other may traverse congested areas during peak-hours.

For each integer valued point $(x,y)$, we calculate the time required to travel a unit distance at a time $t$ by
\begin{equation}\label{eq:td_travel_time}
\tau_{(x,y)} (t) = \frac{\tau_{(x,y)}^\mathrm{free}}{1 - \min\big\{0.8, \gamma(x,y)\big\} \cdot \delta(t) },
\end{equation}
where $\tau_{(x,y)}^\mathrm{free}$ is the congestion-agnostic time necessary to travel a unit distance at point ${(x,y)}$.
The factor $\min\big\{0.8, \gamma(x,y)\big\}$ in (\ref{eq:td_travel_time}) is used to generate homogeneous travel times in an area around the city centres.
The maximum travel time close to a city centre is five times the free flow travel time.

\begin{figure}[htbp]
      \centering
      \begin{subfigure}[b]{0.32\textwidth}
        \includegraphics[width=5cm]{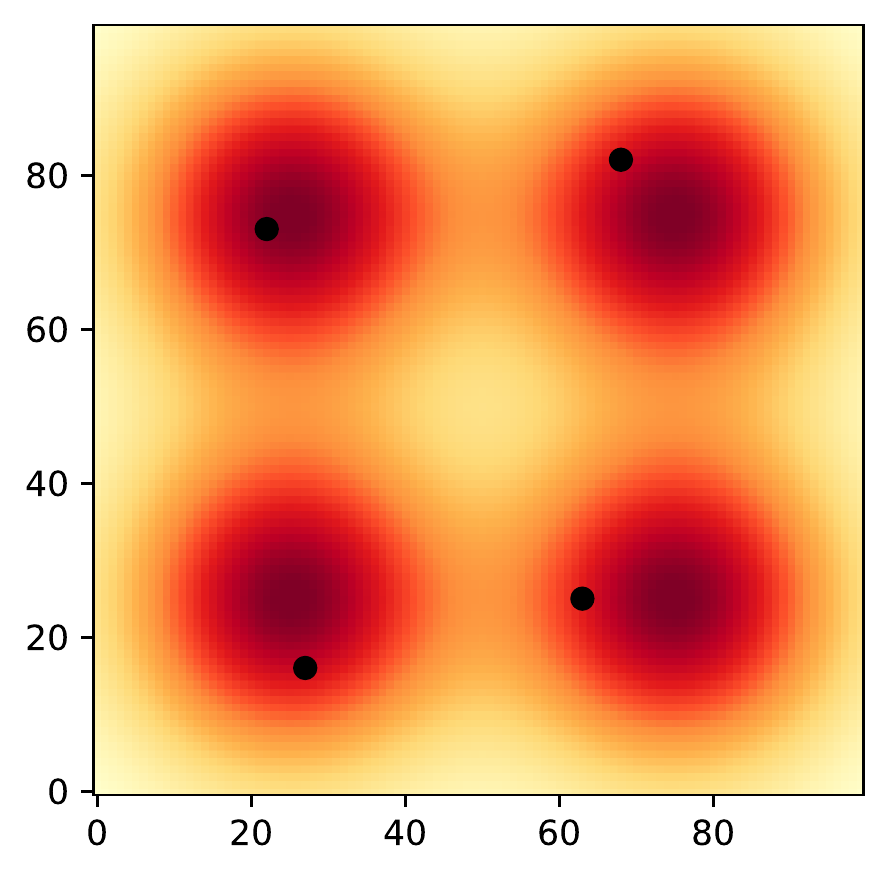}
        \subcaption{One station per city}
      \end{subfigure}
      \begin{subfigure}[b]{0.32\textwidth}
        \includegraphics[width=5cm]{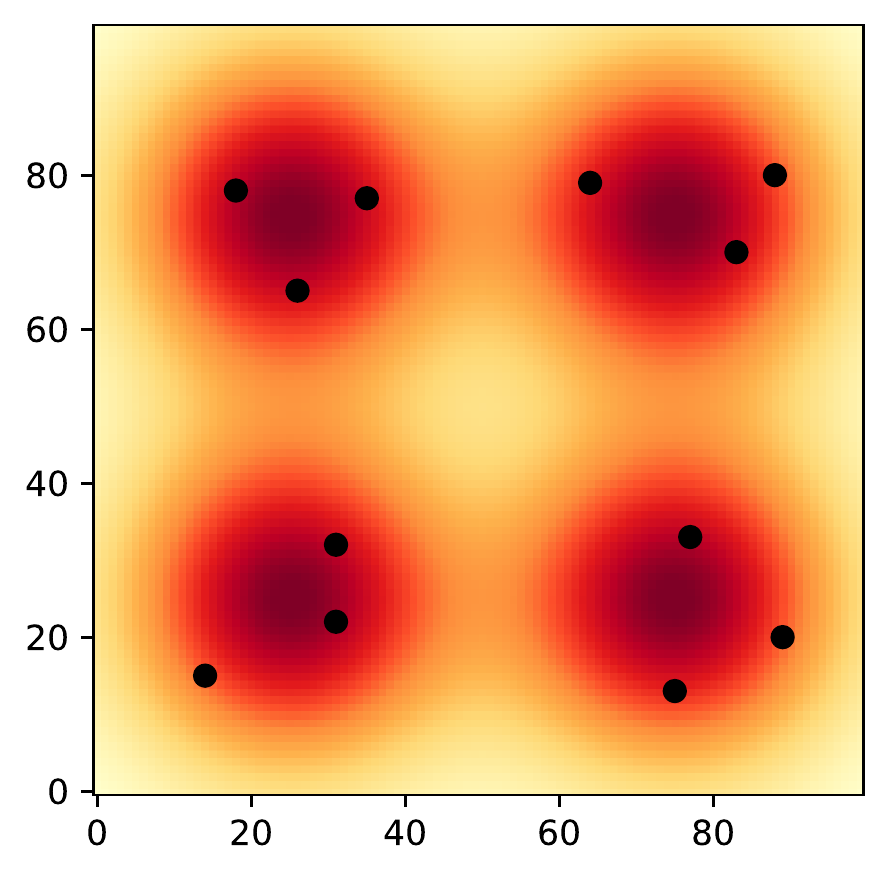}
        \subcaption{Three stations per city}
      \end{subfigure}
      \begin{subfigure}[b]{0.32\textwidth}
        \includegraphics[width=5cm]{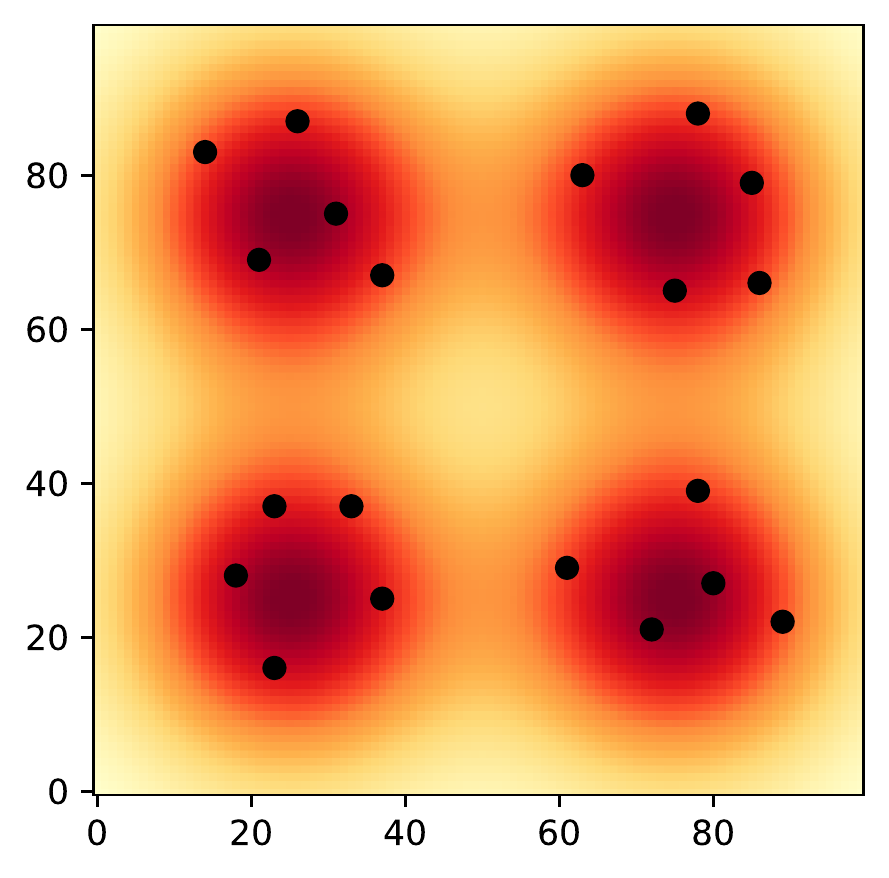}
        \subcaption{Five stations per city}
      \end{subfigure}
\caption{Location of the charging stations}\label{fig:chargingstations}
\end{figure}
For our experiments, we derived five sets of instances from the original instances.
In the first set of instances, vehicles cannot recharge the battery during the route.
In the second set of instances, vehicles can only recharge the battery at the depot.
In the other set of instances, vehicles can recharge the battery at the depot and at public charging infrastructure close to one of the four city centres.
The number of charging stations available ranges from one per city to five per city and their locations are illustrated in Figure~\ref{fig:chargingstations}.

In order to determine time-dependent travel time and consumption functions for any pair of locations, we determine a time-dependent shortest path for different start times distributed over the planning horizon, assuming that the vehicle can only move one unit left, right, up, or down, and $\sqrt{2}$ units in diagonal direction.
A time-dependent travel time function is then determined by fitting a piecewise linear function to the duration of the time-dependent shortest paths.
For each of the time-dependent shortest paths, we determine the energy consumed along the path by accumulating the speed-dependent energy consumption \citep{Galvin2017} for each point traversed along the path.
Then, the time-dependent consumption function is obtained by fitting a piecewise linear function to the energy consumption of the time-dependent shortest paths.
The details of the instance generator including the source code can be found online at \url{https://github.com/SteffenPottel/td_vrptw_instancegenerator}.

The instances are heuristically solved using an adaptation of the savings algorithm \citep{CW1964},
where routes are evaluated using the approaches presented in this paper.
Since vehicles and the corresponding driver are the main cost-components in last-mile delivery operations, we optimize over the total number of vehicles first and, as a second criterion, over the total completion time.
The latter includes, in particular, the time to return to the depot.
In order to allow recharging of the battery at the depot or at public charging stations we allow vehicles to take a detour via a charging station whenever the increase in the travel distance does not exceed  $\sqrt{2}$ times the direct distance between delivery locations.
If multiple charging stations satisfy this criterion, then the charging station with the smallest detour is chosen.
The delivery companies partnering in the project expressed that they want to avoid short and possibly avoidable stops at charging stations.
By setting the duration of the replenishment to a constant value we penalize recharging when the  energy level of the battery is still high.

We implemented our algorithms in C++ and ran the experiments on a single core of an Intel Core i7-7700HQ CPU with 2.80GHz.
For comparison purposes we ran experiments using CPLEX for solving for the mixed integer program provided in Appendix~\ref{app:MIP}.

Tables~\ref{tab:type1Instances} and~\ref{tab:type2Instances} provide an overview over the results of our experiments.
Detailed results of our experiments can be found in Appendix~\ref{sec:detailed_results}.
The first column in the tables indicates the solver used for route evaluations.
The dynamic discretization discovery approach is denoted by DDD, the dynamic discretization discovery approach with preloading of vertices is denoted by DDD-PL, and the approach using the mixed integer programming solver is denoted by MIP.
The second column lists the number of charging stations available.
The third and fourth column show the average computation time (in seconds) until termination of the savings algorithm, and the computation time (in seconds) per route evaluation.
The fifth and six column report the average number of vehicles and the average total completion time.
The next two columns indicate the number of routes in which the vehicle recharges the battery at a service station (or the depot) and the total number of replenishments.
The last column indicates the number of instances for which the savings algorithm terminated within the run time limit of 7200 seconds.
\begin{table}[htbp]
\centering
\scalebox{.67}{
\begin{tabular}{ccccccccc}
\hline
Solver & Charging & Avg. CPU & CPU per Evaluation & Avg. Veh. & Avg. Compl. & Replenishments & Routes w. Repl. & Terminated \\
\hline
DDD & none & 82.4 & 0.000182 & 18.0 & 5949.4 & 0 & 0 & 29 \\
 DDD & Depot & 93.9 & 0.000161 & 18.0 & 5948.2 & 6 & 6 & 29 \\
 DDD & Depot + 1 per city & 134.8 & 0.000203 & 18.1 & 5958.3 & 9 & 9 & 29 \\
 DDD & Depot + 3 per city & 135.3 & 0.000186 & 18.0 & 5954.8 & 10 & 10 & 29 \\
 DDD & Depot + 5 per city & 146.9 & 0.000198 & 18.0 & 5960.5 & 10 & 10 & 29 \\
 DDD-PL & none & 72.6 & 0.000162 & 18.0 & 5949.4 & 0 & 0 & 29 \\
 DDD-PL & Depot & 75.8 & 0.000128 & 18.0 & 5948.2 & 6 & 6 & 29 \\
 DDD-PL & Depot + 1 per city & 84.0 & 0.000127 & 18.1 & 5958.3 & 9 & 9 & 29 \\
 DDD-PL & Depot + 3 per city & 81.1 & 0.000113 & 18.0 & 5954.8 & 10 & 10 & 29 \\
 DDD-PL & Depot + 5 per city & 88.4 & 0.000121 & 18.0 & 5960.5 & 10 & 10 & 29 \\
 MIP & none & 2508.0 & 0.005232 & 18.1 & 5950.6 & 0 & 0 & 29 \\
 MIP & Depot & 3299.9 & 0.005505 & 18.4 & 5752.5 & 4 & 4 & 28 \\
 MIP & Depot + 1 per city & 3621.5 & 0.005620 & 18.4 & 5753.5 & 7 & 7 & 28 \\
 MIP & Depot + 3 per city & 3345.9 & 0.005407 & 19.3 & 5622.5 & 5 & 5 & 23 \\
 MIP & Depot + 5 per city & 3282.6 & 0.005342 & 19.5 & 5753.8 & 6 & 6 & 22 \\
 \hline
\end{tabular}
}
\caption{Results averaged over 29 instances of type C1, R1, and RC1.}\label{tab:type1Instances}
\end{table}

\begin{table}[htbp]
\centering
\scalebox{.67}{
\begin{tabular}{ccccccccc}
\hline
Solver & Charging & Avg. CPU & CPU per Evaluation & Avg. Veh. & Avg. Compl. & Replenishments & Routes w. Repl. & Terminated \\
\hline
DDD & none & 718.2 & 0.001021 & 9.6 & 8386.9 & 0 & 0 & 27 \\
 DDD & Depot & 934.5 & 0.000802 & 6.7 & 6936.7 & 188 & 124 & 27 \\
 DDD & Depot + 1 per city & 1254.9 & 0.000853 & 6.8 & 6884.8 & 190 & 122 & 27 \\
 DDD & Depot + 3 per city & 2021.6 & 0.000739 & 6.2 & 6430.8 & 189 & 113 & 26 \\
 DDD & Depot + 5 per city & 2306.4 & 0.000761 & 6.3 & 6858.1 & 172 & 105 & 23 \\
 DDD-PL & none & 721.2 & 0.001020 & 9.6 & 8386.9 & 0 & 0 & 27 \\
 DDD-PL & Depot & 640.4 & 0.000552 & 6.7 & 6936.7 & 188 & 124 & 27 \\
 DDD-PL & Depot + 1 per city & 739.2 & 0.000513 & 6.8 & 6884.8 & 190 & 122 & 27 \\
 DDD-PL & Depot + 3 per city & 1030.3 & 0.000351 & 6.2 & 6654.4 & 197 & 117 & 27 \\
 DDD-PL & Depot + 5 per city & 1045.7 & 0.000305 & 6.2 & 6854.9 & 194 & 116 & 26 \\
 MIP & none & 3729.6 & 0.005853 & 10.0 & 9056.6 & 0 & 0 & 21 \\
 MIP & Depot & 4877.2 & 0.005443 & 7.1 & 8100.1 & 55 & 40 & 8 \\
 MIP & Depot + 1 per city & 5542.7 & 0.005701 & 6.2 & 7588.0 & 29 & 25 & 5 \\
 MIP & Depot + 3 per city & 6725.1 & 0.005562 & 4.0 & 10778.2 & 1 & 1 & 1 \\
 MIP & Depot + 5 per city & 6230.4 & 0.005494 & 4.0 & 10746.7 & 3 & 3 & 1 \\
 \hline
\end{tabular}
}
\caption{Results averaged over 27 instances of type C2, R2, and RC2.}\label{tab:type2Instances}
\end{table}

As we can see in Tables~\ref{tab:type1Instances} and~\ref{tab:type2Instances},
the dynamic discretization discovery approaches are magnitudes faster than the savings algorithm using the MIP formulation for route evaluations which fails to terminate for many of the instances within the run time limit of 7200~seconds.
We have to note, that although the route evaluations are embedded within the same savings algorithm, the MIP approach
can produce results that differ from those of the DDD approaches, because the MIP formulation is not restricted to the discretization.
Small deviations in calculating the savings can result in different rankings of the savings and ultimately in different routes being generated.

Instances belonging to sets C1, R1, and RC1, have a smaller freight capacity and a shorter planning horizon compared to instances  belonging to sets C2, R2, and RC2.
Therefore, more vehicles are used for instances of sets C1, R1, and RC1 than for instances of sets C2, R2, and RC2.
Thus routes for instances in sets C1, R1, and RC1 visit fewer customers and the approach requires significantly less computational effort than for instances of sets C2, R2, and RC2.

With fewer stops per route, the likelihood of running out of electric energy is much smaller.
In fact, the solutions for instances of sets R1 and RC1 have no recharging stops at all, because these instances  have a very  short planning horizon, leaving little time to sensibly recharge batteries.
For instances in sets C2, R2, and RC2, the freight capacity and planning horizons are larger.
Therefore, less vehicles are used and many more routes make use  of the possibility to recharge the batteries.

With an increasing number of charging stations, the number of alternative routes via charging stations increases and so does the computational effort required for evaluating these alternatives.
We can see that preloading of vertices helps stabilizing the runtime of the algorithm and leads to significantly shorter computation times for instances with many charging possibilities.
With preloading of vertices, the dynamic discretization discovery approach is able to (heuristically) solve almost all of the instances except for instance r211 with five charging stations per city within the run time limit.
Without run time limit, instance r211 with five charging stations per city can be solved in 8027.5~seconds.

An interesting observation from an application point of view is that the possibility to recharge the battery during a route can significantly reduce the number of vehicles required as well as the total completion time.
Table~\ref{tab:type2Instances} indicates that for instances  in sets C2, R2, and RC2, the number of routes required can be reduced by approximately one third if vehicles can recharge batteries while they are on route.
This demonstrates the importance of considering replenishments within scheduling.
Furthermore, it can be seen that the number of vehicles required is not significantly reduced with an increasing number of charging stations.
In fact, the possibility to recharge the battery at the depot already accounts for the main benefit in terms of the number of vehicles required.
Additional public charging infrastructure may furthermore reduce the detour required to reach a charging station and thus can contribute to smaller total completion times and less routes requiring replenishments.

\section{Conclusions}

In this paper we introduce a time-dependent activity scheduling problem in which activities consume a limited resource during execution and activity durations and resource consumptions are time-dependent.
Moreover, we study the case in which  the resource can be replenished between conducting subsequent activities.
We propose a dynamic discretization discovery algorithm which is based on partially time-expanded networks which are dynamically filled with additional vertices.
The dynamic discretization discovery algorithm can be used for general duration and consumption functions and only requires the first-in-first-out property for activity durations.
For the case that the dynamic discretization discovery is embedded in an iterative solution procedure that frequently evaluates activity sequences that start with the same activities, we propose to preload the partially expanded network with vertices generated in previous iterations.
This preloading of vertices can significantly reduce the computational effort required.

We evaluate our approaches on a case of determining routes for a fleet of electric delivery vehicles for last-mile deliveries.
Our experimental results indicate that the dynamic discretization discovery algorithm is magnitudes faster than the commercial solver CPLEX using a mixed integer programming formulation of the problem.
Furthermore, our experiments show that for our instances, which are motivated by the real-life case of last-mile delivery operations, the possibility to recharge batteries en-route can significantly reduce the number of vehicles required and the total completion time.

\section*{Acknowledgements}
This research is supported by the German Research Foundation (DFG) under grant GO 1841/5-1 and by the Federal Ministry of Transport and Digital Infrastructure (BMVI) for the project ZUKUNFT.DE (funding code 03EMF0101K).

\bibliographystyle{ormsv080}
\bibliography{literature}

\clearpage
\section*{Appendix}
\appendix


\section{Mixed-integer program for piecewise linear duration and consumption functions}\label{app:MIP}

In this appendix we provide a formulation of the TDASPR under the assumption that functions $\tau_i(t)$, $\theta_i(t)$, and $\rho_i(t)$ are non-negative, piecewise linear, and lower semi-continuous for each $1\leq i \leq n$
and that the replenishment duration $\Delta_i$ is a constant  for each $1\leq i \leq n$.
Like before we assume that the FIFO property holds, i.e., that $\theta_i(t)$ is non-decreasing.
Note that we do not require $\tau_i(t)$, $\theta_i(t)$, and $\rho_i(t)$ to be continuous.
This allows us to consider applications such as durations of trips conducted with public transport services where discontinuities occur when the scheduled departure time is missed and waiting time until the next scheduled departure needs to be considered.

Let $S^\tau_i$ and $S^\rho_i$ denote the set of piecewise linear segments of $\tau_i$ and $\rho_i$.
For each linear segment $s\in S^\tau_i \cup S^\rho_i$, let $e_s$ and $l_s$ denote the start and end of the segment and let $a_s$ denote the slope and $b_s$ the intercept of the respective linear function.
Thus, if the start time $t_i$ of activity $i$ is within segment $s\in S^\tau_i$, then $\tau_{i}(t_i) = a_s t_i + b_s$.

With this, the problem can be modelled with binary variables
$x_{i,s}$ indicating whether segment $s \in S^\tau_{i}$ of the duration function or segment $s \in S^\rho_{i}$ of the consumption function is used for activity~$i$.
A binary variable $y_{i}$ indicates whether the resource is replenished after conducting the activity~$i$,
linear variables $d_i$ and $q_i$ indicate the duration and resource consumption of activity~$i$,
and $t_{i}$ indicates the start time of activity~$i$.
Linear variables $t_{i,s}$ indicate time values for activity~$i$ associated to segment $s\in  S^\tau_i$ of the duration function or for segment $s\in  S^\rho_i$ of the consumption function.

The TDASP with piecewise linear duration and consumption functions can be formulated as a mixed-integer program by

minimise
\begin{equation}\label{constr:completion}
t_n + d_n
\end{equation}
subject to
\begin{subequations}
\begin{equation}\label{constr:onetimesegment}
\sum_{s \in S^\tau_{i}} x_{i,s} = 1 \forall 1\leq i \leq n
\end{equation}
\begin{equation}\label{constr:righttimesegment}
e_s x_{i,s} \leq  t_{i,s} \leq l_s x_{i,s} \forall 1\leq i \leq n, s \in S^\tau_{i}
\end{equation}
\begin{equation}\label{constr:starttau}
t_i = \sum_{s \in S^\tau_{i}}  t_{i,s} \forall 1\leq i \leq n
\end{equation}
\begin{equation}\label{constr:duration}
d_i = \sum_{s \in S^\tau_i} (a_s t_{i,s} + b_s x_{i,s}) \forall 1\leq i \leq n
\end{equation}
\begin{equation}\label{constr:timing}
t_i + d_i + \Delta_i y_i \leq t_{i+1}  \forall 1\leq i < n
\end{equation}

\end{subequations}
\begin{subequations}
\begin{equation}\label{constr:oneconsumptionsegment}
\sum_{s \in S^\rho_{i}} x_{i,s} = 1 \forall 1\leq i \leq n
\end{equation}
\begin{equation}\label{constr:rightconsumptionsegment}
e_s x_{i,s} \leq  t_{i,s} \leq l_s x_{i,s} \forall 1\leq i \leq n, s \in S^\rho_{i}
\end{equation}
\begin{equation}\label{constr:startrho}
t_i = \sum_{s \in S^\rho_{i}}  t_{i,s} \forall 1\leq i \leq n
\end{equation}
\begin{equation}\label{constr:consumption}
q_i = \sum_{s \in S^\rho_i} (a_s t_{i,s} + b_s x_{i,s}) \forall 1\leq i \leq n
\end{equation}
\begin{equation}\label{constr:capacity}
\sum_{k=i}^j q_k \leq Q + M \sum_{k=i}^{j-1} y_k \forall 1\leq i \leq j \leq n
\end{equation}

\end{subequations}
\begin{equation}\label{constr:x}
x_{i,s} \in \{0,1\}  \forall s \in S^\tau_i \cup S^\rho_i, 1 \leq i  \leq n
\end{equation}
\begin{equation}\label{constr:y}
y_{i} \in \{0,1\}  \forall 1 \leq i  \leq n
\end{equation}

The objective (\ref{constr:completion}) is to minimise the completion time.
Constraint (\ref{constr:onetimesegment}) ensures that for each activity exactly one of the time segments is selected.
Constraints (\ref{constr:righttimesegment}) and (\ref{constr:starttau}) ensure that $t_{i,s}$ is zero if $x_{i,s}=0$ and  $t_{i,s}\in [e_s,l_s]$ if $x_{i,s}=1$, and that the start time of activity~$i$ is set accordingly.
Note that the boundary values of $[e_s,l_s]$ might not belong to the segment $s$ because we do not require continuity of the functions.
Due to the lower semi-continuity, however, this is not a problem in an optimal solution.
The duration $d_i$ of activity $i$ can be obtained by constraint (\ref{constr:duration}).
Because of constraint (\ref{constr:timing}), the start time of any activity must not be before the completion time of the previous activity plus the time required for a potential replenishment.
Constraints (\ref{constr:oneconsumptionsegment}) to (\ref{constr:consumption}) are analogous to constraints (\ref{constr:onetimesegment}) to (\ref{constr:duration}) for the resource consumption.
Constraint (\ref{constr:capacity}) requires that the cumulative resource consumption of any sequence of activities does not exceed $Q$ unless the resource is replenished.
Lastly, the domain of the binary variables is given by (\ref{constr:x}) and (\ref{constr:y}).

\section{Restrictions on replenishments}\label{app:restrictedreplenishments}
A variation of the TDASPR, in which replenishments are forbidden or required after certain activities, can be obtained by constraining the binary variables indicating the choice of whether to conduct a replenishment or not.
This can be achieved by adding a constraint $y_i=0$ or $y_i=1$ for each activity~$i$ which prohibits or requires a replenishment after its completion in the problems given by (\ref{constr:completion:replenishments}) to (\ref{constr:domains:replenishments})
or (\ref{constr:completion}) to (\ref{constr:y}).

If a replenishment is not allowed after an activity, we can avoid replenishments by adapting Algorithm~\ref{alg:replenishments:A*} in such a way that the earliest time after the replenishment $t^*$ is set to a value that exceeds~$l_{i+1}$ in line~\ref{line:earliestreplenishmenttime}.

Let us now consider the case that a replenishment is required after activity~$i$.
In this case Property~\ref{prop:successor} loses its relevance for all vertices of activity $i$, because the property does not consider the duration of the required replenishment.
The main purpose of Property~\ref{prop:successor} was to ensure that the vertex representing the earliest possible time at which the next activity can be conducted is included in the partially expanded network.
In Algorithm~\ref{alg:replenishments:A*}, this is achieved by line~\ref{line:earliestreplenishment}  which includes a vertex for the next activity and the earliest time after the required replenishment.
We can adapt Algorithm~\ref{alg:replenishments:A*} in such a way that lines~\ref{line:noreplenishmentloop} to~\ref{line:noreplenishmentloopend}, which update the labels assuming that no replenishment is conducted, are skipped if a replenishment is required.
Furthermore, the condition in line~\ref{line:recurse} of function ${\tt DDD\!::\!addRecursive}$ can be changed in such a way that the recursive insertion of vertices is terminated if a replenishment is required after the current activity.
With these changes our approach can be used to solve the TDASPR with required replenishments even if Property~\ref{prop:successor} does not hold for activities requiring a replenishment.


\section{Detailed results}\label{sec:detailed_results}

Tables~\ref{tab:0-service} to ~\ref{tab:21-service} show detailed results of our computational experiments.
The first column gives the name of the instance.
The results for the approach based on dynamic discretization discovery with replenishments (DDD) are given in columns 2 to 5.
The results for the dynamic discretization discover approach with replenishments and preloading (DDD-PL) are given in columns 6 to 9.
The results for the approach using CPLEX for solving for the mixed integer program (MIP) are given in columns 10 to 13.
For each approach, the total computation time is given in the column titled \emph{CPU},
the number of vehicles required is given in the column titled \emph{Veh.},
the total completion time is given in the column titled \emph{Compl.},
and the number of replenishments at service stations is given in the column titled \emph{Repl.}.
Dashes in the columns for MIP indicate that no solution was found within the time limit of 7200~seconds.
For the dynamic discretization discovery approaches DDD and DDD-PL, the results violating the time limit are written in italic.

\begin{table}[htb]
\centering
\renewcommand{\arraystretch}{0.7}
{\scriptsize

\begin{tabular}{lrrrrrrrrrrrr}
\hline
& DDD &&&& DDD-PL &&&& MIP &&& \\
\hline
Instance & CPU & Veh. & Compl. & Repl. & CPU & Veh. & Compl. & Repl. & CPU & Veh. & Compl. & Repl. \\
\hline
c101 & 101.4 & 13 & 11589.1 & 0 & 97.8 & 13 & 11589.1 & 0 & 1083.0 & 14 & 11507.9 & 0\\
c102 & 645.3 & 16 & 12064.9 & 0 & 628.2 & 16 & 12064.9 & 0 & 2189.1 & 15 & 11974.3 & 0\\
c103 & 271.2 & 11 & 11809.4 & 0 & 256.5 & 11 & 11809.4 & 0 & 3308.8 & 12 & 11978.7 & 0\\
c104 & 60.2 & 11 & 11788.7 & 0 & 49.2 & 11 & 11788.7 & 0 & 4389.2 & 11 & 11644.2 & 0\\
c105 & 96.4 & 15 & 10946.8 & 0 & 89.1 & 15 & 10946.8 & 0 & 1505.1 & 14 & 10930.6 & 0\\
c106 & 116.9 & 16 & 12357.3 & 0 & 110.3 & 16 & 12357.3 & 0 & 1601.9 & 16 & 12432.2 & 0\\
c107 & 126.7 & 13 & 10653.1 & 0 & 115.6 & 13 & 10653.1 & 0 & 2054.6 & 12 & 10487.1 & 0\\
c108 & 139.0 & 13 & 10884.0 & 0 & 122.6 & 13 & 10884.0 & 0 & 2351.3 & 15 & 11077.0 & 0\\
c109 & 160.4 & 13 & 10612.8 & 0 & 143.1 & 13 & 10612.8 & 0 & 3379.3 & 13 & 10606.5 & 0\\
c201 & 29.4 & 5 & 11321.1 & 0 & 7.3 & 5 & 11321.1 & 0 & 2054.1 & 5 & 11317.0 & 0\\
c202 & 417.9 & 8 & 17256.2 & 0 & 382.6 & 8 & 17256.2 & 0 & 2863.3 & 9 & 17647.4 & 0\\
c203 & 143.3 & 9 & 17404.2 & 0 & 120.2 & 9 & 17404.2 & 0 & 4360.3 & 11 & 18616.1 & 0\\
c204 & 4599.1 & 8 & 16319.8 & 0 & 4804.2 & 8 & 16319.8 & 0 & --- & --- & --- & ---\\
c205 & 47.6 & 5 & 11123.6 & 0 & 16.2 & 5 & 11123.6 & 0 & 2662.4 & 5 & 11119.2 & 0\\
c206 & 171.2 & 8 & 15428.4 & 0 & 126.0 & 8 & 15428.4 & 0 & 2902.9 & 8 & 15554.2 & 0\\
c207 & 248.5 & 10 & 19200.6 & 0 & 202.9 & 10 & 19200.6 & 0 & 2825.2 & 10 & 18998.1 & 0\\
c208 & 97.1 & 5 & 10839.4 & 0 & 42.9 & 5 & 10839.4 & 0 & 5048.7 & 5 & 10834.5 & 0\\
r101 & 21.2 & 28 & 4596.6 & 0 & 18.5 & 28 & 4596.6 & 0 & 473.6 & 28 & 4619.5 & 0\\
r102 & 37.3 & 23 & 3864.5 & 0 & 31.3 & 23 & 3864.5 & 0 & 1767.6 & 23 & 3835.1 & 0\\
r103 & 33.7 & 19 & 3470.5 & 0 & 24.4 & 19 & 3470.5 & 0 & 2853.6 & 19 & 3491.1 & 0\\
r104 & 25.4 & 15 & 2951.3 & 0 & 12.0 & 15 & 2951.3 & 0 & 4224.7 & 14 & 2839.2 & 0\\
r105 & 44.9 & 22 & 3717.0 & 0 & 41.4 & 22 & 3717.0 & 0 & 852.4 & 23 & 3913.1 & 0\\
r106 & 38.0 & 21 & 3373.2 & 0 & 29.4 & 21 & 3373.2 & 0 & 2228.1 & 21 & 3332.1 & 0\\
r107 & 36.0 & 19 & 3093.4 & 0 & 25.7 & 19 & 3093.4 & 0 & 3216.4 & 20 & 3241.5 & 0\\
r108 & 24.6 & 16 & 2937.1 & 0 & 9.9 & 16 & 2937.1 & 0 & 4398.2 & 16 & 2917.9 & 0\\
r109 & 56.4 & 20 & 3273.5 & 0 & 48.4 & 20 & 3273.5 & 0 & 1811.9 & 19 & 3185.3 & 0\\
r110 & 37.7 & 18 & 2964.8 & 0 & 26.9 & 18 & 2964.8 & 0 & 2859.1 & 18 & 3036.1 & 0\\
r111 & 33.4 & 18 & 3126.4 & 0 & 21.3 & 18 & 3126.4 & 0 & 3153.0 & 18 & 3079.2 & 0\\
r112 & 28.2 & 15 & 2638.0 & 0 & 9.2 & 15 & 2638.0 & 0 & 4878.3 & 15 & 2632.5 & 0\\
r201 & 254.3 & 12 & 6799.5 & 0 & 246.2 & 12 & 6799.5 & 0 & 2502.3 & 12 & 6839.2 & 0\\
r202 & 604.6 & 11 & 6422.3 & 0 & 592.1 & 11 & 6422.3 & 0 & 3366.1 & 11 & 6204.6 & 0\\
r203 & 415.7 & 9 & 4988.8 & 0 & 420.5 & 9 & 4988.8 & 0 & 4085.9 & 9 & 4880.5 & 0\\
r204 & 783.1 & 8 & 4210.2 & 0 & 797.1 & 8 & 4210.2 & 0 & --- & --- & --- & ---\\
r205 & 493.0 & 10 & 5732.0 & 0 & 491.3 & 10 & 5732.0 & 0 & 3678.4 & 11 & 6038.5 & 0\\
r206 & 1064.1 & 10 & 5211.6 & 0 & 1155.0 & 10 & 5211.6 & 0 & 4797.2 & 10 & 5341.0 & 0\\
r207 & 814.1 & 9 & 4700.4 & 0 & 844.3 & 9 & 4700.4 & 0 & 5726.8 & 9 & 4453.6 & 0\\
r208 & 810.9 & 8 & 3528.3 & 0 & 820.8 & 8 & 3528.3 & 0 & --- & --- & --- & ---\\
r209 & 785.4 & 11 & 5820.9 & 0 & 784.0 & 11 & 5820.9 & 0 & 4621.4 & 11 & 5785.2 & 0\\
r210 & 536.3 & 10 & 5400.0 & 0 & 543.4 & 10 & 5400.0 & 0 & 4554.3 & 11 & 5888.0 & 0\\
r211 & 1463.2 & 9 & 4499.5 & 0 & 1501.0 & 9 & 4499.5 & 0 & --- & --- & --- & ---\\
rc101 & 26.5 & 25 & 4553.7 & 0 & 24.5 & 25 & 4553.7 & 0 & 668.5 & 24 & 4400.1 & 0\\
rc102 & 34.5 & 23 & 4007.6 & 0 & 28.9 & 23 & 4007.6 & 0 & 1702.7 & 23 & 3933.6 & 0\\
rc103 & 31.0 & 21 & 3739.8 & 0 & 22.4 & 21 & 3739.8 & 0 & 2748.0 & 21 & 3774.4 & 0\\
rc104 & 30.2 & 15 & 3123.7 & 0 & 16.0 & 15 & 3123.7 & 0 & 4073.1 & 16 & 3198.8 & 0\\
rc105 & 36.9 & 22 & 4030.0 & 0 & 32.8 & 22 & 4030.0 & 0 & 1348.8 & 24 & 4142.6 & 0\\
rc106 & 47.9 & 22 & 3735.4 & 0 & 42.8 & 22 & 3735.4 & 0 & 1360.1 & 21 & 3661.8 & 0\\
rc107 & 24.8 & 22 & 3515.4 & 0 & 16.2 & 22 & 3515.4 & 0 & 2488.5 & 22 & 3504.7 & 0\\
rc108 & 23.7 & 17 & 3115.6 & 0 & 10.1 & 17 & 3115.6 & 0 & 3763.3 & 18 & 3191.4 & 0\\
rc201 & 479.4 & 14 & 7851.4 & 0 & 470.6 & 14 & 7851.4 & 0 & 2118.6 & 14 & 7992.2 & 0\\
rc202 & 346.8 & 13 & 7266.6 & 0 & 345.7 & 13 & 7266.6 & 0 & 3475.4 & 12 & 7203.4 & 0\\
rc203 & 417.6 & 11 & 5670.8 & 0 & 400.1 & 11 & 5670.8 & 0 & 4632.7 & 11 & 5723.4 & 0\\
rc204 & 941.1 & 9 & 4235.3 & 0 & 906.5 & 9 & 4235.3 & 0 & --- & --- & --- & ---\\
rc205 & 495.4 & 13 & 7177.7 & 0 & 489.4 & 13 & 7177.7 & 0 & 3102.3 & 13 & 7081.8 & 0\\
rc206 & 772.5 & 12 & 6883.9 & 0 & 768.8 & 12 & 6883.9 & 0 & 4046.7 & 12 & 6880.8 & 0\\
rc207 & 761.0 & 12 & 6388.6 & 0 & 755.6 & 12 & 6388.6 & 0 & 4897.0 & 12 & 5790.4 & 0\\
rc208 & 1399.6 & 11 & 4764.5 & 0 & 1438.3 & 11 & 4764.5 & 0 & --- & --- & --- & ---\\
\hline
\end{tabular}
\caption{Results for instances without service stations.}\label{tab:0-service}
}
\end{table}

\begin{table}[htb]
\centering
\renewcommand{\arraystretch}{0.7}
{\scriptsize

\begin{tabular}{lrrrrrrrrrrrr}
\hline
& DDD &&&& DDD-PL &&&& MIP &&& \\
\hline
Instance & CPU & Veh. & Compl. & Repl. & CPU & Veh. & Compl. & Repl. & CPU & Veh. & Compl. & Repl. \\
\hline
c101 & 112.7 & 13 & 11589.1 & 0 & 105.7 & 13 & 11589.1 & 0 & 1570.7 & 14 & 11507.9 & 0\\
c102 & 662.5 & 17 & 12063.2 & 1 & 650.8 & 17 & 12063.2 & 1 & 3620.6 & 15 & 11974.3 & 0\\
c103 & 296.2 & 11 & 12038.8 & 2 & 274.7 & 11 & 12038.8 & 2 & 5572.4 & 12 & 12213.4 & 2\\
c104 & 80.9 & 11 & 11590.4 & 1 & 56.8 & 11 & 11590.4 & 1 & --- & --- & --- & ---\\
c105 & 105.9 & 15 & 10946.8 & 0 & 92.0 & 15 & 10946.8 & 0 & 2194.6 & 14 & 10930.6 & 0\\
c106 & 124.4 & 16 & 12292.3 & 2 & 114.3 & 16 & 12292.3 & 2 & 2367.1 & 16 & 12343.1 & 2\\
c107 & 138.6 & 13 & 10653.1 & 0 & 120.5 & 13 & 10653.1 & 0 & 2999.0 & 12 & 10487.1 & 0\\
c108 & 151.1 & 13 & 10884.0 & 0 & 127.1 & 13 & 10884.0 & 0 & 3583.9 & 15 & 11077.0 & 0\\
c109 & 179.7 & 13 & 10612.8 & 0 & 147.9 & 13 & 10612.8 & 0 & 5342.9 & 13 & 10606.5 & 0\\
c201 & 78.1 & 4 & 10800.4 & 2 & 16.1 & 4 & 10800.4 & 2 & 3493.4 & 4 & 10864.3 & 2\\
c202 & 166.5 & 8 & 13213.7 & 11 & 33.7 & 8 & 13213.7 & 11 & 5350.8 & 9 & 14985.1 & 9\\
c203 & 947.5 & 6 & 14026.1 & 8 & 700.4 & 6 & 14026.1 & 8 & --- & --- & --- & ---\\
c204 & 829.3 & 5 & 13141.1 & 8 & 671.6 & 5 & 13141.1 & 8 & --- & --- & --- & ---\\
c205 & 119.8 & 4 & 10609.8 & 2 & 36.4 & 4 & 10609.8 & 2 & 4953.0 & 4 & 10576.7 & 2\\
c206 & 481.9 & 5 & 12078.3 & 8 & 125.9 & 5 & 12078.3 & 8 & --- & --- & --- & ---\\
c207 & 742.5 & 5 & 13219.8 & 11 & 465.2 & 5 & 13219.8 & 11 & --- & --- & --- & ---\\
c208 & 278.8 & 4 & 10388.0 & 2 & 140.0 & 4 & 10388.0 & 2 & --- & --- & --- & ---\\
r101 & 21.4 & 28 & 4596.6 & 0 & 18.9 & 28 & 4596.6 & 0 & 486.3 & 28 & 4619.5 & 0\\
r102 & 41.0 & 23 & 3864.5 & 0 & 32.5 & 23 & 3864.5 & 0 & 2332.4 & 23 & 3835.1 & 0\\
r103 & 44.3 & 19 & 3470.5 & 0 & 26.4 & 19 & 3470.5 & 0 & 4156.0 & 19 & 3491.1 & 0\\
r104 & 42.7 & 15 & 2951.3 & 0 & 14.5 & 15 & 2951.3 & 0 & 6477.7 & 14 & 2839.2 & 0\\
r105 & 46.0 & 22 & 3717.0 & 0 & 41.5 & 22 & 3717.0 & 0 & 884.5 & 23 & 3913.1 & 0\\
r106 & 42.9 & 21 & 3373.2 & 0 & 30.2 & 21 & 3373.2 & 0 & 2896.1 & 21 & 3332.1 & 0\\
r107 & 46.4 & 19 & 3093.4 & 0 & 27.5 & 19 & 3093.4 & 0 & 4605.1 & 20 & 3241.5 & 0\\
r108 & 42.1 & 16 & 2937.1 & 0 & 12.4 & 16 & 2937.1 & 0 & 6614.1 & 16 & 2917.9 & 0\\
r109 & 57.3 & 20 & 3273.5 & 0 & 48.9 & 20 & 3273.5 & 0 & 1912.1 & 19 & 3185.3 & 0\\
r110 & 40.5 & 18 & 2964.8 & 0 & 27.5 & 18 & 2964.8 & 0 & 3259.8 & 18 & 3036.1 & 0\\
r111 & 38.8 & 18 & 3126.4 & 0 & 22.7 & 18 & 3126.4 & 0 & 3862.2 & 18 & 3079.2 & 0\\
r112 & 33.3 & 15 & 2638.0 & 0 & 10.5 & 15 & 2638.0 & 0 & 5568.3 & 15 & 2632.5 & 0\\
r201 & 244.7 & 7 & 5173.2 & 10 & 226.5 & 7 & 5173.2 & 10 & 3869.1 & 7 & 5197.1 & 8\\
r202 & 1636.8 & 9 & 5846.2 & 8 & 1294.1 & 9 & 5846.2 & 8 & 6243.5 & 8 & 5509.5 & 7\\
r203 & 587.8 & 8 & 4510.7 & 5 & 408.5 & 8 & 4510.7 & 5 & --- & --- & --- & ---\\
r204 & 307.2 & 5 & 3675.6 & 6 & 173.5 & 5 & 3675.6 & 6 & --- & --- & --- & ---\\
r205 & 1229.0 & 8 & 4897.6 & 7 & 795.3 & 8 & 4897.6 & 7 & --- & --- & --- & ---\\
r206 & 1649.9 & 7 & 4529.1 & 7 & 1147.0 & 7 & 4529.1 & 7 & --- & --- & --- & ---\\
r207 & 758.8 & 7 & 3606.0 & 4 & 467.2 & 7 & 3606.0 & 4 & --- & --- & --- & ---\\
r208 & 919.6 & 5 & 3417.8 & 6 & 474.3 & 5 & 3417.8 & 6 & --- & --- & --- & ---\\
r209 & 1403.9 & 7 & 4583.4 & 8 & 935.3 & 7 & 4583.4 & 8 & --- & --- & --- & ---\\
r210 & 2123.2 & 7 & 4276.1 & 6 & 1503.8 & 7 & 4276.1 & 6 & --- & --- & --- & ---\\
r211 & 2597.0 & 6 & 3616.2 & 5 & 1894.3 & 6 & 3616.2 & 5 & --- & --- & --- & ---\\
rc101 & 35.5 & 25 & 4553.7 & 0 & 24.0 & 25 & 4553.7 & 0 & 694.4 & 24 & 4400.1 & 0\\
rc102 & 48.3 & 23 & 4007.6 & 0 & 28.9 & 23 & 4007.6 & 0 & 2059.9 & 23 & 3933.6 & 0\\
rc103 & 50.0 & 21 & 3739.8 & 0 & 22.8 & 21 & 3739.8 & 0 & 3635.7 & 21 & 3774.4 & 0\\
rc104 & 56.5 & 15 & 3123.7 & 0 & 17.5 & 15 & 3123.7 & 0 & 5767.0 & 16 & 3198.8 & 0\\
rc105 & 49.4 & 22 & 4030.0 & 0 & 32.5 & 22 & 4030.0 & 0 & 1452.3 & 24 & 4142.6 & 0\\
rc106 & 64.1 & 22 & 3735.4 & 0 & 42.1 & 22 & 3735.4 & 0 & 1437.7 & 21 & 3661.8 & 0\\
rc107 & 35.7 & 22 & 3515.4 & 0 & 16.3 & 22 & 3515.4 & 0 & 2801.4 & 22 & 3504.7 & 0\\
rc108 & 35.6 & 17 & 3115.6 & 0 & 10.4 & 17 & 3115.6 & 0 & 4243.9 & 18 & 3191.4 & 0\\
rc201 & 1068.9 & 11 & 6859.0 & 7 & 760.0 & 11 & 6859.0 & 7 & 3792.7 & 8 & 6058.0 & 8\\
rc202 & 614.7 & 8 & 5630.5 & 9 & 435.3 & 8 & 5630.5 & 9 & 5983.0 & 8 & 5627.6 & 9\\
rc203 & 484.1 & 9 & 4591.2 & 7 & 312.3 & 9 & 4591.2 & 7 & --- & --- & --- & ---\\
rc204 & 1307.1 & 5 & 3924.7 & 6 & 940.8 & 5 & 3924.7 & 6 & --- & --- & --- & ---\\
rc205 & 919.5 & 10 & 6271.4 & 10 & 667.1 & 10 & 6271.4 & 10 & 5332.3 & 9 & 5982.3 & 10\\
rc206 & 1288.5 & 7 & 5019.8 & 9 & 893.5 & 7 & 5019.8 & 9 & --- & --- & --- & ---\\
rc207 & 831.2 & 9 & 5249.0 & 9 & 575.2 & 9 & 5249.0 & 9 & --- & --- & --- & ---\\
rc208 & 1614.4 & 6 & 4136.9 & 7 & 1196.6 & 6 & 4136.9 & 7 & --- & --- & --- & ---\\
\hline
\end{tabular}
\caption{Results for instances with recharging at the depot.}\label{tab:1-service}
}
\end{table}

\begin{table}[htb]
\centering
\renewcommand{\arraystretch}{0.7}
{\scriptsize

\begin{tabular}{lrrrrrrrrrrrr}
\hline
& DDD &&&& DDD-PL &&&& MIP &&& \\
\hline
Instance & CPU & Veh. & Compl. & Repl. & CPU & Veh. & Compl. & Repl. & CPU & Veh. & Compl. & Repl. \\
\hline
c101 & 216.7 & 13 & 11589.1 & 0 & 148.3 & 13 & 11589.1 & 0 & 1839.8 & 14 & 11507.9 & 0\\
c102 & 898.6 & 17 & 12063.2 & 1 & 663.6 & 17 & 12063.2 & 1 & 4255.9 & 15 & 12042.2 & 2\\
c103 & 629.0 & 13 & 12129.1 & 4 & 416.1 & 13 & 12129.1 & 4 & 6915.7 & 12 & 12219.9 & 2\\
c104 & 122.2 & 11 & 11719.3 & 2 & 60.2 & 11 & 11719.3 & 2 & --- & --- & --- & ---\\
c105 & 144.1 & 15 & 10946.8 & 0 & 94.9 & 15 & 10946.8 & 0 & 2538.3 & 14 & 10930.6 & 0\\
c106 & 185.9 & 16 & 12365.6 & 2 & 126.3 & 16 & 12365.6 & 2 & 3064.5 & 16 & 12297.4 & 3\\
c107 & 189.1 & 13 & 10653.1 & 0 & 123.8 & 13 & 10653.1 & 0 & 3476.7 & 12 & 10487.1 & 0\\
c108 & 214.9 & 13 & 10884.0 & 0 & 134.2 & 13 & 10884.0 & 0 & 4384.8 & 15 & 11077.0 & 0\\
c109 & 249.4 & 13 & 10612.8 & 0 & 155.8 & 13 & 10612.8 & 0 & 6372.4 & 13 & 10606.5 & 0\\
c201 & 149.3 & 4 & 10754.0 & 2 & 19.7 & 4 & 10754.0 & 2 & 4363.4 & 4 & 10880.6 & 2\\
c202 & 799.7 & 8 & 13114.4 & 10 & 211.8 & 8 & 13114.4 & 10 & --- & --- & --- & ---\\
c203 & 2078.1 & 6 & 14079.3 & 7 & 1006.9 & 6 & 14079.3 & 7 & --- & --- & --- & ---\\
c204 & 3071.9 & 5 & 13153.1 & 8 & 1559.3 & 5 & 13153.1 & 8 & --- & --- & --- & ---\\
c205 & 232.9 & 4 & 10622.9 & 2 & 38.4 & 4 & 10622.9 & 2 & 6382.4 & 4 & 10656.9 & 2\\
c206 & 932.9 & 5 & 12398.6 & 9 & 335.1 & 5 & 12398.6 & 9 & --- & --- & --- & ---\\
c207 & 1361.8 & 6 & 13238.2 & 10 & 331.9 & 6 & 13238.2 & 10 & --- & --- & --- & ---\\
c208 & 496.5 & 4 & 10388.0 & 2 & 161.3 & 4 & 10388.0 & 2 & --- & --- & --- & ---\\
r101 & 28.2 & 28 & 4596.6 & 0 & 18.9 & 28 & 4596.6 & 0 & 488.0 & 28 & 4619.5 & 0\\
r102 & 54.7 & 23 & 3864.5 & 0 & 32.6 & 23 & 3864.5 & 0 & 2384.4 & 23 & 3835.1 & 0\\
r103 & 59.2 & 19 & 3470.5 & 0 & 26.4 & 19 & 3470.5 & 0 & 4282.0 & 19 & 3491.1 & 0\\
r104 & 57.7 & 15 & 2951.3 & 0 & 14.9 & 15 & 2951.3 & 0 & 6679.9 & 14 & 2839.2 & 0\\
r105 & 60.1 & 22 & 3717.0 & 0 & 42.1 & 22 & 3717.0 & 0 & 880.3 & 23 & 3913.1 & 0\\
r106 & 56.9 & 21 & 3373.2 & 0 & 30.8 & 21 & 3373.2 & 0 & 2952.0 & 21 & 3332.1 & 0\\
r107 & 62.6 & 19 & 3093.4 & 0 & 27.8 & 19 & 3093.4 & 0 & 4969.0 & 20 & 3241.5 & 0\\
r108 & 56.8 & 16 & 2937.1 & 0 & 12.9 & 16 & 2937.1 & 0 & 6827.5 & 16 & 2917.9 & 0\\
r109 & 75.5 & 20 & 3273.5 & 0 & 48.7 & 20 & 3273.5 & 0 & 1916.4 & 19 & 3185.3 & 0\\
r110 & 54.1 & 18 & 2964.8 & 0 & 27.5 & 18 & 2964.8 & 0 & 3299.7 & 18 & 3036.1 & 0\\
r111 & 51.6 & 18 & 3126.4 & 0 & 22.8 & 18 & 3126.4 & 0 & 3975.8 & 18 & 3079.2 & 0\\
r112 & 44.8 & 15 & 2638.0 & 0 & 10.5 & 15 & 2638.0 & 0 & 5667.4 & 15 & 2632.5 & 0\\
r201 & 362.1 & 8 & 5228.0 & 8 & 253.5 & 8 & 5228.0 & 8 & 4074.7 & 7 & 5237.2 & 9\\
r202 & 1705.6 & 8 & 5562.3 & 7 & 1269.4 & 8 & 5562.3 & 7 & 6423.5 & 8 & 5408.0 & 7\\
r203 & 561.2 & 8 & 4493.1 & 5 & 390.5 & 8 & 4493.1 & 5 & --- & --- & --- & ---\\
r204 & 849.3 & 5 & 3735.3 & 6 & 510.2 & 5 & 3735.3 & 6 & --- & --- & --- & ---\\
r205 & 1721.7 & 7 & 4660.5 & 7 & 1159.9 & 7 & 4660.5 & 7 & --- & --- & --- & ---\\
r206 & 2038.9 & 7 & 4570.1 & 7 & 1455.0 & 7 & 4570.1 & 7 & --- & --- & --- & ---\\
r207 & 669.9 & 7 & 3586.8 & 4 & 373.7 & 7 & 3586.8 & 4 & --- & --- & --- & ---\\
r208 & 996.6 & 5 & 3303.6 & 5 & 539.1 & 5 & 3303.6 & 5 & --- & --- & --- & ---\\
r209 & 1829.5 & 8 & 4613.9 & 7 & 1216.3 & 8 & 4613.9 & 7 & --- & --- & --- & ---\\
r210 & 913.2 & 8 & 4598.5 & 7 & 621.5 & 8 & 4598.5 & 7 & --- & --- & --- & ---\\
r211 & 2836.2 & 6 & 3612.4 & 5 & 2059.4 & 6 & 3612.4 & 5 & --- & --- & --- & ---\\
rc101 & 35.7 & 25 & 4553.7 & 0 & 24.1 & 25 & 4553.7 & 0 & 707.4 & 24 & 4400.1 & 0\\
rc102 & 50.4 & 23 & 4007.6 & 0 & 29.2 & 23 & 4007.6 & 0 & 2307.7 & 23 & 3933.6 & 0\\
rc103 & 54.1 & 21 & 3739.8 & 0 & 23.5 & 21 & 3739.8 & 0 & 4111.0 & 21 & 3774.4 & 0\\
rc104 & 64.3 & 15 & 3123.7 & 0 & 18.4 & 15 & 3123.7 & 0 & 6503.2 & 16 & 3198.8 & 0\\
rc105 & 50.5 & 22 & 4030.0 & 0 & 32.8 & 22 & 4030.0 & 0 & 1548.1 & 24 & 4142.6 & 0\\
rc106 & 65.3 & 22 & 3735.4 & 0 & 42.1 & 22 & 3735.4 & 0 & 1513.0 & 21 & 3661.8 & 0\\
rc107 & 37.1 & 22 & 3515.4 & 0 & 16.5 & 22 & 3515.4 & 0 & 2953.0 & 22 & 3504.7 & 0\\
rc108 & 39.1 & 17 & 3115.6 & 0 & 11.0 & 17 & 3115.6 & 0 & 4588.9 & 18 & 3191.4 & 0\\
rc201 & 500.4 & 9 & 6035.4 & 10 & 313.2 & 9 & 6035.4 & 10 & 6469.3 & 8 & 5757.2 & 9\\
rc202 & 506.7 & 9 & 5530.7 & 11 & 342.1 & 9 & 5530.7 & 11 & --- & --- & --- & ---\\
rc203 & 484.8 & 8 & 4406.0 & 8 & 317.1 & 8 & 4406.0 & 8 & --- & --- & --- & ---\\
rc204 & 501.9 & 6 & 4009.1 & 6 & 287.5 & 6 & 4009.1 & 6 & --- & --- & --- & ---\\
rc205 & 1207.3 & 9 & 5812.7 & 12 & 818.9 & 9 & 5812.7 & 12 & --- & --- & --- & ---\\
rc206 & 2284.8 & 8 & 5272.5 & 9 & 1426.5 & 8 & 5272.5 & 9 & --- & --- & --- & ---\\
rc207 & 1412.4 & 8 & 5032.1 & 10 & 902.4 & 8 & 5032.1 & 10 & --- & --- & --- & ---\\
rc208 & 3377.5 & 7 & 4078.1 & 6 & 2038.5 & 7 & 4078.1 & 6 & --- & --- & --- & ---\\
\hline
\end{tabular}
\caption{Results for instances with recharging at the depot and one station per city.}\label{tab:5-service}
}
\end{table}

\begin{table}[htb]
\centering
\renewcommand{\arraystretch}{0.7}
{\scriptsize

\begin{tabular}{lrrrrrrrrrrrr}
\hline
& DDD &&&& DDD-PL &&&& MIP &&& \\
\hline
Instance & CPU & Veh. & Compl. & Repl. & CPU & Veh. & Compl. & Repl. & CPU & Veh. & Compl. & Repl. \\
\hline
c101 & 169.5 & 13 & 11589.1 & 0 & 115.3 & 13 & 11589.1 & 0 & 2059.8 & 14 & 11507.9 & 0\\
c102 & 930.9 & 17 & 12063.2 & 1 & 681.4 & 17 & 12063.2 & 1 & 4734.5 & 15 & 12079.1 & 1\\
c103 & 432.2 & 11 & 11992.5 & 3 & 292.6 & 11 & 11992.5 & 3 & --- & --- & --- & ---\\
c104 & 152.9 & 11 & 11732.0 & 3 & 63.6 & 11 & 11732.0 & 3 & --- & --- & --- & ---\\
c105 & 156.8 & 15 & 10946.8 & 0 & 99.1 & 15 & 10946.8 & 0 & 2851.8 & 14 & 10930.6 & 0\\
c106 & 202.5 & 16 & 12389.2 & 3 & 130.6 & 16 & 12389.2 & 3 & 3611.1 & 16 & 12262.3 & 4\\
c107 & 202.1 & 13 & 10653.1 & 0 & 128.3 & 13 & 10653.1 & 0 & 3928.6 & 12 & 10487.1 & 0\\
c108 & 246.3 & 13 & 10884.0 & 0 & 149.1 & 13 & 10884.0 & 0 & 5407.3 & 15 & 11077.0 & 0\\
c109 & 272.4 & 13 & 10612.8 & 0 & 161.3 & 13 & 10612.8 & 0 & --- & --- & --- & ---\\
c201 & 414.4 & 4 & 10758.8 & 2 & 31.1 & 4 & 10758.8 & 2 & 6725.1 & 4 & 10778.2 & 1\\
c202 & 1162.7 & 7 & 12369.5 & 9 & 206.6 & 7 & 12369.5 & 9 & --- & --- & --- & ---\\
c203 & 2520.5 & 6 & 14003.2 & 8 & 579.7 & 6 & 14003.2 & 8 & --- & --- & --- & ---\\
c204 & 4710.0 & 5 & 13259.3 & 7 & 1262.7 & 5 & 13259.3 & 7 & --- & --- & --- & ---\\
c205 & 677.5 & 4 & 10071.8 & 2 & 63.4 & 4 & 10071.8 & 2 & --- & --- & --- & ---\\
c206 & 1413.5 & 5 & 12160.6 & 8 & 171.0 & 5 & 12160.6 & 8 & --- & --- & --- & ---\\
c207 & {\it 12235.4} & {\it 5} & {\it 12469.8} & {\it 8} & 6787.1 & 5 & 12469.8 & 8 & --- & --- & --- & ---\\
c208 & 2069.5 & 4 & 10040.9 & 2 & 612.1 & 4 & 10040.9 & 2 & --- & --- & --- & ---\\
r101 & 28.5 & 28 & 4596.6 & 0 & 18.9 & 28 & 4596.6 & 0 & 504.2 & 28 & 4619.5 & 0\\
r102 & 58.1 & 23 & 3864.5 & 0 & 33.3 & 23 & 3864.5 & 0 & 2754.6 & 23 & 3835.1 & 0\\
r103 & 67.7 & 19 & 3470.5 & 0 & 27.8 & 19 & 3470.5 & 0 & 5128.7 & 19 & 3491.1 & 0\\
r104 & 68.4 & 15 & 2951.3 & 0 & 16.4 & 15 & 2951.3 & 0 & --- & --- & --- & ---\\
r105 & 60.8 & 22 & 3717.0 & 0 & 41.8 & 22 & 3717.0 & 0 & 931.2 & 23 & 3913.1 & 0\\
r106 & 61.2 & 21 & 3373.2 & 0 & 31.5 & 21 & 3373.2 & 0 & 3440.5 & 21 & 3332.1 & 0\\
r107 & 71.3 & 19 & 3093.4 & 0 & 28.8 & 19 & 3093.4 & 0 & 5651.0 & 20 & 3241.5 & 0\\
r108 & 68.0 & 16 & 2937.1 & 0 & 14.2 & 16 & 2937.1 & 0 & --- & --- & --- & ---\\
r109 & 79.1 & 20 & 3273.5 & 0 & 50.2 & 20 & 3273.5 & 0 & 2070.6 & 19 & 3185.3 & 0\\
r110 & 58.9 & 18 & 2964.8 & 0 & 28.5 & 18 & 2964.8 & 0 & 3690.6 & 18 & 3036.1 & 0\\
r111 & 59.8 & 18 & 3126.4 & 0 & 23.4 & 18 & 3126.4 & 0 & 4616.6 & 18 & 3079.2 & 0\\
r112 & 56.5 & 15 & 2638.0 & 0 & 11.8 & 15 & 2638.0 & 0 & 6645.8 & 15 & 2632.5 & 0\\
r201 & 716.5 & 7 & 5192.1 & 9 & 352.9 & 7 & 5192.1 & 9 & --- & --- & --- & ---\\
r202 & 1038.9 & 8 & 5185.6 & 10 & 688.2 & 8 & 5185.6 & 10 & --- & --- & --- & ---\\
r203 & 903.3 & 8 & 4053.7 & 7 & 535.2 & 8 & 4053.7 & 7 & --- & --- & --- & ---\\
r204 & 324.3 & 4 & 3698.4 & 5 & 143.1 & 4 & 3698.4 & 5 & --- & --- & --- & ---\\
r205 & 4064.5 & 7 & 4353.7 & 8 & 1197.2 & 7 & 4353.7 & 8 & --- & --- & --- & ---\\
r206 & 2180.7 & 6 & 4183.6 & 8 & 1021.8 & 6 & 4183.6 & 8 & --- & --- & --- & ---\\
r207 & 6237.2 & 6 & 3756.8 & 5 & 983.5 & 6 & 3756.8 & 5 & --- & --- & --- & ---\\
r208 & 1360.3 & 5 & 3478.2 & 6 & 570.5 & 5 & 3478.2 & 6 & --- & --- & --- & ---\\
r209 & 2220.9 & 7 & 4082.1 & 7 & 1357.0 & 7 & 4082.1 & 7 & --- & --- & --- & ---\\
r210 & 1814.5 & 7 & 4512.7 & 7 & 1036.7 & 7 & 4512.7 & 7 & --- & --- & --- & ---\\
r211 & 6747.9 & 5 & 3571.9 & 6 & 3528.3 & 5 & 3571.9 & 6 & --- & --- & --- & ---\\
rc101 & 36.0 & 25 & 4553.7 & 0 & 25.0 & 25 & 4553.7 & 0 & 727.9 & 24 & 4400.1 & 0\\
rc102 & 52.1 & 23 & 4007.6 & 0 & 29.5 & 23 & 4007.6 & 0 & 2439.6 & 23 & 3933.6 & 0\\
rc103 & 58.6 & 21 & 3739.8 & 0 & 24.7 & 21 & 3739.8 & 0 & 4504.1 & 21 & 3774.4 & 0\\
rc104 & 71.9 & 15 & 3123.7 & 0 & 19.6 & 15 & 3123.7 & 0 & --- & --- & --- & ---\\
rc105 & 51.6 & 22 & 4030.0 & 0 & 33.7 & 22 & 4030.0 & 0 & 1613.9 & 24 & 4142.6 & 0\\
rc106 & 65.4 & 22 & 3735.4 & 0 & 43.3 & 22 & 3735.4 & 0 & 1538.2 & 21 & 3661.8 & 0\\
rc107 & 39.9 & 22 & 3515.4 & 0 & 17.1 & 22 & 3515.4 & 0 & 3154.1 & 22 & 3504.7 & 0\\
rc108 & 43.6 & 17 & 3115.6 & 0 & 11.7 & 17 & 3115.6 & 0 & 4950.6 & 18 & 3191.4 & 0\\
rc201 & 559.2 & 8 & 5729.6 & 12 & 310.7 & 8 & 5729.6 & 12 & --- & --- & --- & ---\\
rc202 & 633.2 & 8 & 5248.6 & 10 & 415.8 & 8 & 5248.6 & 10 & --- & --- & --- & ---\\
rc203 & 645.4 & 8 & 4484.0 & 7 & 409.1 & 8 & 4484.0 & 7 & --- & --- & --- & ---\\
rc204 & 634.5 & 5 & 3835.8 & 6 & 350.5 & 5 & 3835.8 & 6 & --- & --- & --- & ---\\
rc205 & 1166.8 & 8 & 5418.7 & 12 & 762.6 & 8 & 5418.7 & 12 & --- & --- & --- & ---\\
rc206 & 2217.4 & 7 & 4900.0 & 9 & 1035.5 & 7 & 4900.0 & 9 & --- & --- & --- & ---\\
rc207 & 1763.2 & 7 & 4601.0 & 9 & 969.5 & 7 & 4601.0 & 9 & --- & --- & --- & ---\\
rc208 & 4365.7 & 6 & 4249.1 & 8 & 2435.1 & 6 & 4249.1 & 8 & --- & --- & --- & ---\\
\hline
\end{tabular}
\caption{Results for instances with recharging at the depot and three stations per city.}\label{tab:13-service}
}
\end{table}

\begin{table}[htb]
\centering
\renewcommand{\arraystretch}{0.7}
{\scriptsize

\begin{tabular}{lrrrrrrrrrrrr}
\hline
& DDD &&&& DDD-PL &&&& MIP &&& \\
\hline
Instance & CPU & Veh. & Compl. & Repl. & CPU & Veh. & Compl. & Repl. & CPU & Veh. & Compl. & Repl. \\
\hline
c101 & 163.7 & 13 & 11589.1 & 0 & 109.4 & 13 & 11589.1 & 0 & 2075.3 & 14 & 11507.9 & 0\\
c102 & 1131.3 & 17 & 12092.0 & 1 & 819.4 & 17 & 12092.0 & 1 & 4747.5 & 15 & 12042.2 & 2\\
c103 & 474.6 & 11 & 11996.9 & 3 & 319.4 & 11 & 11996.9 & 3 & --- & --- & --- & ---\\
c104 & 157.1 & 11 & 11734.8 & 3 & 65.6 & 11 & 11734.8 & 3 & --- & --- & --- & ---\\
c105 & 157.6 & 15 & 10946.8 & 0 & 102.0 & 15 & 10946.8 & 0 & 2830.1 & 14 & 10930.6 & 0\\
c106 & 203.1 & 16 & 12518.0 & 3 & 136.7 & 16 & 12518.0 & 3 & 3492.2 & 16 & 12198.2 & 4\\
c107 & 204.3 & 13 & 10653.1 & 0 & 131.1 & 13 & 10653.1 & 0 & 3925.5 & 12 & 10487.1 & 0\\
c108 & 280.3 & 13 & 10884.0 & 0 & 172.2 & 13 & 10884.0 & 0 & 5678.6 & 15 & 11077.0 & 0\\
c109 & 277.5 & 13 & 10612.8 & 0 & 167.8 & 13 & 10612.8 & 0 & --- & --- & --- & ---\\
c201 & 575.9 & 4 & 10778.6 & 3 & 42.8 & 4 & 10778.6 & 3 & 6230.4 & 4 & 10746.7 & 3\\
c202 & 1288.0 & 7 & 13123.4 & 9 & 178.5 & 7 & 13123.4 & 9 & --- & --- & --- & ---\\
c203 & 4930.8 & 6 & 14079.3 & 9 & 1421.5 & 6 & 14079.3 & 9 & --- & --- & --- & ---\\
c204 & 2570.6 & 6 & 14850.6 & 7 & 1425.2 & 6 & 14850.6 & 7 & --- & --- & --- & ---\\
c205 & 774.3 & 4 & 10224.1 & 3 & 77.2 & 4 & 10224.1 & 3 & --- & --- & --- & ---\\
c206 & 3835.5 & 5 & 12270.4 & 8 & 1221.1 & 5 & 12270.4 & 8 & --- & --- & --- & ---\\
c207 & {\it 8697.8} & {\it 5} & {\it 12641.5} & {\it 10} & 1124.0 & 5 & 12641.5 & 10 & --- & --- & --- & ---\\
c208 & 1783.7 & 4 & 10091.8 & 2 & 347.6 & 4 & 10091.8 & 2 & --- & --- & --- & ---\\
r101 & 28.8 & 28 & 4596.6 & 0 & 18.7 & 28 & 4596.6 & 0 & 513.5 & 28 & 4619.5 & 0\\
r102 & 59.6 & 23 & 3864.5 & 0 & 33.7 & 23 & 3864.5 & 0 & 2821.4 & 23 & 3835.1 & 0\\
r103 & 68.9 & 19 & 3470.5 & 0 & 28.3 & 19 & 3470.5 & 0 & 5237.5 & 19 & 3491.1 & 0\\
r104 & 71.6 & 15 & 2951.3 & 0 & 16.3 & 15 & 2951.3 & 0 & --- & --- & --- & ---\\
r105 & 61.3 & 22 & 3717.0 & 0 & 41.9 & 22 & 3717.0 & 0 & 956.4 & 23 & 3913.1 & 0\\
r106 & 62.9 & 21 & 3373.2 & 0 & 32.0 & 21 & 3373.2 & 0 & 3517.4 & 21 & 3332.1 & 0\\
r107 & 73.4 & 19 & 3093.4 & 0 & 29.0 & 19 & 3093.4 & 0 & 5790.4 & 20 & 3241.5 & 0\\
r108 & 71.6 & 16 & 2937.1 & 0 & 14.4 & 16 & 2937.1 & 0 & --- & --- & --- & ---\\
r109 & 87.4 & 20 & 3273.5 & 0 & 56.3 & 20 & 3273.5 & 0 & 2191.8 & 19 & 3185.3 & 0\\
r110 & 62.3 & 18 & 2964.8 & 0 & 29.4 & 18 & 2964.8 & 0 & 3970.4 & 18 & 3036.1 & 0\\
r111 & 61.8 & 18 & 3126.4 & 0 & 24.3 & 18 & 3126.4 & 0 & 4841.8 & 18 & 3079.2 & 0\\
r112 & 68.7 & 15 & 2638.0 & 0 & 12.9 & 15 & 2638.0 & 0 & --- & --- & --- & ---\\
r201 & 439.5 & 7 & 5183.7 & 10 & 188.9 & 7 & 5183.7 & 10 & --- & --- & --- & ---\\
r202 & 1228.2 & 7 & 5076.0 & 9 & 827.9 & 7 & 5076.0 & 9 & --- & --- & --- & ---\\
r203 & 655.5 & 8 & 4053.2 & 8 & 415.9 & 8 & 4053.2 & 8 & --- & --- & --- & ---\\
r204 & 435.3 & 5 & 3602.7 & 5 & 223.5 & 5 & 3602.7 & 5 & --- & --- & --- & ---\\
r205 & 5509.8 & 6 & 4306.3 & 8 & 1222.2 & 6 & 4306.3 & 8 & --- & --- & --- & ---\\
r206 & {\it 7646.0} & {\it 7} & {\it 4328.1} & {\it 7} & 2577.0 & 7 & 4328.1 & 7 & --- & --- & --- & ---\\
r207 & {\it 14042.4} & {\it 6} & {\it 3522.8} & {\it 5} & 2011.1 & 6 & 3522.8 & 5 & --- & --- & --- & ---\\
r208 & 1836.2 & 5 & 3418.4 & 6 & 518.6 & 5 & 3418.4 & 6 & --- & --- & --- & ---\\
r209 & 2253.6 & 6 & 4304.3 & 8 & 1080.7 & 6 & 4304.3 & 8 & --- & --- & --- & ---\\
r210 & 6243.7 & 7 & 4274.3 & 7 & 2822.7 & 7 & 4274.3 & 7 & --- & --- & --- & ---\\
r211 & {\it 28924.2} & {\it 5} & {\it 3508.6} & {\it 5} & {\it 8027.5} & {\it 5} & {\it 3508.6} & {\it 5} & --- & --- & --- & ---\\
rc101 & 36.2 & 25 & 4553.7 & 0 & 24.5 & 25 & 4553.7 & 0 & 738.7 & 24 & 4400.1 & 0\\
rc102 & 53.1 & 23 & 4007.6 & 0 & 29.7 & 23 & 4007.6 & 0 & 2517.9 & 23 & 3933.6 & 0\\
rc103 & 60.3 & 21 & 3739.8 & 0 & 24.2 & 21 & 3739.8 & 0 & 4609.4 & 21 & 3774.4 & 0\\
rc104 & 75.4 & 15 & 3123.7 & 0 & 19.5 & 15 & 3123.7 & 0 & --- & --- & --- & ---\\
rc105 & 52.1 & 22 & 4030.0 & 0 & 32.9 & 22 & 4030.0 & 0 & 1677.6 & 24 & 4142.6 & 0\\
rc106 & 67.3 & 22 & 3735.4 & 0 & 43.2 & 22 & 3735.4 & 0 & 1632.7 & 21 & 3661.8 & 0\\
rc107 & 41.0 & 22 & 3515.4 & 0 & 17.0 & 22 & 3515.4 & 0 & 3246.2 & 22 & 3504.7 & 0\\
rc108 & 47.3 & 17 & 3115.6 & 0 & 11.8 & 17 & 3115.6 & 0 & 5205.7 & 18 & 3191.4 & 0\\
rc201 & 817.4 & 8 & 5783.4 & 10 & 422.7 & 8 & 5783.4 & 10 & --- & --- & --- & ---\\
rc202 & 816.7 & 7 & 5042.8 & 11 & 461.9 & 7 & 5042.8 & 11 & --- & --- & --- & ---\\
rc203 & 781.4 & 9 & 4393.2 & 7 & 472.8 & 9 & 4393.2 & 7 & --- & --- & --- & ---\\
rc204 & 1351.0 & 5 & 3683.4 & 6 & 590.2 & 5 & 3683.4 & 6 & --- & --- & --- & ---\\
rc205 & 1202.1 & 8 & 5395.7 & 11 & 795.4 & 8 & 5395.7 & 11 & --- & --- & --- & ---\\
rc206 & 5815.9 & 7 & 4973.4 & 9 & 2376.0 & 7 & 4973.4 & 9 & --- & --- & --- & ---\\
rc207 & 2052.8 & 7 & 4784.9 & 8 & 1148.8 & 7 & 4784.9 & 8 & --- & --- & --- & ---\\
rc208 & 5848.9 & 6 & 4041.8 & 8 & 3193.4 & 6 & 4041.8 & 8 & --- & --- & --- & ---\\
\hline
\end{tabular}
\caption{Results for instances with recharging at the depot and five stations per city.}\label{tab:21-service}
}
\end{table}

\end{document}